\documentclass[letterpaper]{article} 
\usepackage{aaai2026}  
\usepackage{times}  
\usepackage{helvet}  
\usepackage{courier}  
\usepackage[hyphens]{url}  
\usepackage{graphicx} 
\usepackage{natbib}  
\usepackage{caption} 
\usepackage{algorithm}
\usepackage[noend]{algpseudocode}
\MakeRobust{\Call}
\makeatletter
\algnewcommand{\LineComment}[1]{\Statex \hskip\ALG@thistlm \(\triangleright\) #1}
\makeatother

\usepackage{newfloat}
\usepackage{listings}
\DeclareCaptionStyle{ruled}{labelfont=normalfont,labelsep=colon,strut=off} 
\floatstyle{ruled}
\newfloat{listing}{tb}{lst}{}
\floatname{listing}{Listing}
\usepackage{amsmath}
\usepackage{amsthm}
\usepackage{amssymb}
\usepackage{mathtools}
\usepackage{thm-restate}

\newtheorem{theorem}{Theorem}

\newtheorem{lemma}[theorem]{Lemma}
\newtheorem{definition}[theorem]{Definition}
\newtheorem{corollary}[theorem]{Corollary}

\usepackage{cleveref}

\renewcommand{\sharp}{\#}

\newcommand{\gphi}{G_{\Phi}}
\newcommand{\tgphi}{\widetilde{G}_{\Phi}}

\newcommand{\smoke}{\texttt{\footnotesize smoke}}
\newcommand{\stress}{\texttt{\footnotesize stress}}
\newcommand{\friend}{\texttt{\footnotesize friend}}

\newcommand{\p}[3]{#1{\big\vert}_{#2\downarrow#3}}

\newcommand{\bin}{\textsf{\tiny bin}}
\newcommand{\num}{\textup{num}}

\newcommand{\fomcbin}{\text{FOMC}_{\bin}}

\newcommand{\ttM}{\texttt{M}}

\newcommand{\cA}{\mathcal{A}}
\newcommand{\cB}{\mathcal{B}}

\newcommand{\cL}{\mathcal{L}}

\newcommand{\cS}{\mathcal{S}}

\newcommand{\va}{\bar{a}}
\newcommand{\vb}{\bar{b}}
\newcommand{\vc}{\bar{c}}

\newcommand{\vt}{\bar{t}}
\newcommand{\vv}{\bar{v}}
\newcommand{\vu}{\bar{u}}
\newcommand{\vw}{\bar{w}}
\newcommand{\vx}{\bar{x}}
\newcommand{\vy}{\bar{y}}
\newcommand{\vz}{\bar{z}}

\newcommand{\bbN}{\mathbb{N}}
\newcommand{\bbB}{\mathbb{B}}

\newcommand{\class}[1]{\textup{#1}}

\newcommand{\nls}{\class{NL}}

\newcommand{\pt}{\class{P}}
\newcommand{\npt}{\class{NP}}

\newcommand{\ps}{\class{PSPACE}}

\newcommand{\expt}{\class{EXP}}
\newcommand{\nexpt}{\class{NEXP}}

\newcommand{\shpt}{\sharp\pt}

\newcommand{\shexpt}{\sharp\expt}

\newcommand{\prob}[1]{\text{\textup{#1}}}

\newcommand{\shtwodqbf}{\sharp\prob{2-DQBF}}

\newcommand{\tool}[1]{\texttt{#1}}

\newcommand{\abs}[1]{{\left\lvert #1\right\rvert}}

\DeclarePairedDelimiterX\Set[1]{\lbrace}{\rbrace}%
{  #1 }
\DeclarePairedDelimiterX{\Paren}[1]{(}{)}{#1}

\newcommand\restr[2]{{
  \kern-\nulldelimiterspace{} 
  #1 
  {\upharpoonright}_{#2} 
  }%
}

\title{Model Counting for Dependency Quantified Boolean Formulas}

\author {
Long-Hin Fung\textsuperscript{\rm 1},
Che Cheng\textsuperscript{\rm 2},
Jie-Hong Roland Jiang\textsuperscript{\rm 2},
Friedrich Slivovsky\textsuperscript{\rm 3},
Tony Tan\textsuperscript{\rm 3}
}
\affiliations {
\textsuperscript{\rm 1}Department of Computer Science and Information Engineering, National Taiwan University \\
\textsuperscript{\rm 2}Graduate Institute of Electronics Engineering, National Taiwan University\\
\textsuperscript{\rm 3}School of Computer Science and Informatics, University of Liverpool\\
r12922017@csie.ntu.edu.tw, 
\{f11943097,jhjiang\}@ntu.edu.tw,
\{F.Slivovsky,tonytan\}@liverpool.ac.uk
}

\begin{document}

\maketitle

\begin{abstract}
Dependency Quantified Boolean Formulas (DQBF) generalize QBF by explicitly specifying which universal variables each existential variable depends on, instead of relying on a linear quantifier order.
The satisfiability problem of DQBF is NEXP-complete, and many hard problems can be succinctly encoded as DQBF.
Recent work has revealed a strong analogy between DQBF and SAT: $k$-DQBF (with $k$ existential variables) is a succinct form of $k$-SAT, and satisfiability is NEXP-complete for $3$-DQBF but PSPACE-complete for $2$-DQBF, mirroring the complexity gap between $3$-SAT (NP-complete) and $2$-SAT (NL-complete).

Motivated by this analogy, we study the model counting problem for DQBF, denoted $\#$DQBF.
Our main theoretical result is that $\#$2-DQBF is $\#$EXP-complete, where $\#$EXP is the exponential-time analogue of $\#$P. This parallels Valiant's classical theorem stating that $\#$2-SAT is $\#$P-complete. As a direct application, we show that first-order model counting (FOMC) remains $\#$EXP-complete even when restricted to a PSPACE-decidable fragment of first-order logic and domain size two.

Building on recent successes in reducing 2-DQBF satisfiability to symbolic model checking, we develop a dedicated 2-DQBF model counter. Using a diverse set of crafted instances, we experimentally evaluated it against a baseline that expands 2-DQBF formulas into propositional formulas and applies propositional model counting. 
While the baseline worked well when each existential variable depends on few variables, our implementation scaled significantly better to larger dependency sets. 
\end{abstract}

\begin{links}
\link{Code and benchmarks}{https://github.com/Sat-DQBF/sharp2DQR}
\link{Arxiv link}{https://arxiv.org/abs/2511.07337}
\end{links}

\section{Introduction}

There has been tremendous progress in SAT solving over the past few decades, enabling widespread applications across many areas of computing, including reasoning tasks in AI~\cite{sat-handbook,sat-museum-23,rev-sat-23}. 
However, certain problems in hardware verification and synthesis are unlikely to admit succinct encodings in propositional logic, prompting research into automated reasoning in more expressive logics~\cite{Jiang09,BalabanovJ15,SchollB01,GitinaRSWSB13,BloemKS14,ChatterjeeHOP13,KuehlmannPKG02,Ge-ErnstSSW22}.

A natural candidate for such applications is the logic of \emph{Dependency Quantified Boolean Formulas} (DQBF), an extension of Quantified Boolean Formulas (QBF) with Henkin quantifiers that annotate each existential variable with a set of universal variables it depends on~\cite{BalabanovCJ14}.
A model of a DQBF consists of \emph{Skolem functions} that map each existential variable to a truth value based on an assignment to its universal dependencies.
The fine-grained control over variable dependencies allows DQBF to naturally express problems such as constrained program synthesis~\cite{GoliaRM21} and equivalence checking of partially specified circuits~\cite{Gitina2013}.
This has led to active research over the past decade
and the development of several solvers~\cite{FrohlichKBV14,TentrupR19,GitinaWRSSB15,WimmerKBS017,SicS21,pedant21,pedant22,GoliaRM23}, as well as the inclusion of a dedicated DQBF track in recent QBF evaluations~\cite{PulinaS19}. 

While satisfiability is the central question in DQBF, 
many synthesis and verification tasks benefit from knowing \emph{how many} solutions exist. 
Counting models can help debug and refine specifications: 
for instance, an unexpectedly large number of Skolem functions may suggest that the specification admits unintended behaviour.
Model counters have been developed for QBF with one quantifier alternation~\cite{PlankMS24} as well as 
Boolean synthesis~\cite{kuldeep-skolem-24}, and more recently, for general QBF~\cite{martina24}.

In this paper, we consider the model counting problem for DQBF, denoted $\sharp$DQBF. This is a formidable problem, since even deciding whether a DQBF has a model is $\nexpt$-complete~\cite{PetersonR79,dqbf-fmcad22,us-sat-25}.
Moreover, because DQBF allows arbitrary and potentially incomparable dependency sets, existing techniques for $\sharp$QBF that rely on a linear order of quantifiers cannot be applied.

To support the intuition that $\sharp$DQBF is a particularly difficult problem, we first prove that even the model counting problem for DQBF with just two existential variables, denoted $\shtwodqbf$, is $\shexpt$-complete.
This is despite the fact that
 satisfiability of $2$-DQBF is ``only'' $\ps$-complete~\cite{dqbf-sat23}.
Our proof builds on a recent result on $2$-DNF model counting~\cite{bannach_p_2025} and uses the close correspondence between $k$-DQBF and $k$-SAT~\cite{dqbf-sat23}.
This hardness result is analogous to the well-known hardness of $\sharp$2-SAT: while $2$-SAT is solvable in polynomial time, counting its models is $\shpt$-complete~\cite{Valiant79a,Valiant79b}.

Note that functions in $\shexpt$ may output doubly exponential numbers, which require exponentially many bits.
Thus, the standard polynomial time Turing reductions for establishing $\shpt$-hardness, as in the case of $\sharp$2-SAT in~\cite{Valiant79a,Valiant79b}, are not appropriate for the class $\shexpt$.
To circumvent this issue, we introduce a new kind of polynomial-time reduction, called a \emph{poly-monious reduction} (see \Cref{sec:prelim} for the definition), which lies between
 classical parsimonious reductions
and polynomial-time Turing reductions.
Under poly-monious reductions,
$\sharp$2-SAT is still $\sharp$P-complete~\cite{bannach_p_2025}. 

Another notion of hardness for the counting problem requires the reduction to be parsimonious~\cite{Ladner89}.
However, since 2-SAT is $\nls$-complete, $\sharp$2-SAT is not $\sharp$P-hard under parsimonious reduction, unless $\nls = \npt$.
We believe that the notion of poly-monious reduction is well-suited for establishing $\shexpt$-hardness, as it strikes a balance between 
parsimonious reductions and polynomial-time 
Turing reductions in terms of strength.

As an application of the hardness of $\shtwodqbf$, we show that the combined complexity of first-order model counting (FOMC)---a central problem in statistical relational AI---is $\shexpt$-complete (over varying vocabulary).
FOMC is defined as given FO sentence $\Psi$ and a number $N$ in unary, compute the number of models of $\Psi$ with domain $\{1,\ldots,N\}$.
The combined complexity of FOMC is the complexity measured in terms of both 
the sentence $\Psi$ and the number $N$.
While the $\sharp\mathrm{EXP}$-hardness of FOMC can already be inferred from classical results in logic~\cite{Lewis80},\footnote{A close inspection of the proof in~\cite{Lewis80} shows poly-monious reductions from languages in $\nexpt$ to the Bernays-Sch\"onfinkel-Ramsey fragment of FO, whose satisfiability problem is known to be $\nexpt$-complete.} we obtain a stronger result: FOMC is $\shexpt$-hard even when the domain size is restricted to $2$ and the base logic is a $\ps$-decidable fragment of FO.\footnote{In general, the satisfiability problem for FO is undecidable~\cite{Trakhtenbrot50}.}
This may help explain why scalable FO model counters have remained elusive despite intensive research efforts for over a decade.

Motivated by our result that $\sharp$DQBF remains hard even for just two existential variables, we explore the viability of solving $\sharp$2-DQBF in practice. 
Due to the double-exponential number of possible Skolem functions, direct enumeration is infeasible. Similarly, expanding a DQBF into a propositional formula leads to an exponential blow-up, rendering state-of-the-art $\sharp$SAT solvers impractical.

Instead, we build on a recent success in reducing 2-DQBF satisfiability to model checking~\cite{dqbf-fmcad24} and interpret 2-DQBF instances as succinctly represented implication graphs.
Based on this idea, we propose a model counting algorithm that proceeds in two main phases. In the first phase, it constructs a Binary Decision Diagram (BDD) representing reachability in the implication graph.
This phase benefits from mature tools developed by the formal methods community, including the IC3 algorithm, the CUDD package for BDD manipulation, and ABC's implementation of exact reachability~\cite{bradley-manna-07,bradley-vmcai-11,een-pdr-fmcad-11,somenzi2009cudd,ABC}.
In the second phase, our algorithm counts Skolem functions by analysing each weakly connected component separately—similar in spirit to component-based decomposition in propositional model counters~\cite{GomesSS21}. Within each component, it suffices to enumerate Skolem functions for just one existential variable. We further restrict attention to \emph{partial} Skolem functions defined only on the variables local to each component. This avoids explicit enumeration and enables us to handle instances with up to $2^{2^{64}}$ Skolem functions. The techniques used in this phase combine new ideas with existing methods~\cite{pedant21,dqbf-fmcad24}.

We evaluate our implementation on a diverse set of crafted benchmarks. As a baseline, we use a pipeline that expands a DQBF to a propositional formula and applies the $\sharp$SAT solver Ganak~\cite{ganak}. While Ganak performs well on some smaller instances, its reliance on explicit expansion becomes a bottleneck as dependency sets grow. In contrast, our solver scales gracefully and consistently outperforms the baseline on DQBF with larger dependency sets.

We also performed experiments with state-of-the-art FO model counters. While our approach can only be applied to FOMC with binary relations, this is enough to encode problems such as counting the number of independent sets in highly symmetric graphs.
In some cases, our implementation was able to handle instances with more than $2^{127}$ solutions, far beyond the practical reach of current FO model counters. This indicates that an analogue of our component decomposition technique for \sharp$2$-DQBF may improve FO model counters in restricted, highly symmetric settings.

\paragraph{Related work.}
FOMC is often studied in the data complexity setting, i.e., the FO sentence is fixed and the complexity is measured only in terms of the domain size.
It is shown in~\cite{BeameBGS15} that there is an $\text{FO}^3$ sentence such that the data complexity of its FOMC is $\shpt_1$-complete.
For the two-variable fragment, the data complexity drops to PTIME~\cite{toth-kuzelka-24,tim-kuzelka-23,BeameBGS15}.
The combined complexity is $\shpt$-complete, but assuming that the vocabulary is fixed~\cite{BeameBGS15}.
A tightly related problem to FOMC is query evaluation on probabilistic databases, whose combined complexity is $\shpt$-complete~\cite{dalvi-suciu-04}, but again, under the assumption of fixed vocabulary.

The notion of combined and data complexity was introduced in~\cite{Vardi82-complexity} in the context of database query evaluation, to better understand which component (the query/the data/both) contributes more to the complexity of query evaluation.
Since then, as hinted in the previous paragraph, it has become the standard notion for establishing fine-grained complexity results for problems involving a few parameters.

\section{Preliminaries}\label{sec:prelim}

\paragraph{Notation.}
Let $\bbB=\{\bot,\top\}$, where $\bot$ and $\top$ denote the Boolean false and true values.
A literal is either a Boolean variable or its negation.
We write
$x^\top$ to denote the literal $x$ and
$x^{\bot}$ to denote $\neg x$.
The sign of the literal $x^b$ is the bit $b$.

We use the symbols $a,b,c$ to denote elements in $\bbB$, and the bar version $\va,\vb,\vc$ 
to denote strings in $\bbB^*$
with $|\va|$ denoting the length of $\va$.
Boolean variables are denoted by $x,y,z,u,v$ 
and the bar version $\vx,\vy,\vz,\vu,\vv$ 
denote vectors of Boolean variables
with $|\vx|$ denoting the length of $\vx$.
We insist that in a vector $\vx$ there is no variable occurring more than once.
Abusing the notation, we write
$\vz\subseteq \vx$ to denote that every variable in $\vz$ also occurs in $\vx$.

As usual, $\varphi(\vx)$ denotes a Boolean formula with variables $\vx$.
When it is clear from the context,
we simply write $\varphi$.
For $\vz\subseteq \vx$ and $\va\in\bbB^*$ where $|\va|=|\vz|$,
$\varphi[\vz/ \va]$ denotes the formula obtained from $\varphi$ by assigning the values in $\va$ to $\vz$.
Obviously, if $\vz=\vx$, then $\varphi[\vz/ \va]$ is either $\bot$ or $\top$.

\paragraph{Poly-monious reductions.}
Let $\Sigma$ be a finite alphabet.
A \emph{poly-monious reduction} from a function $F:\Sigma^*\to\bbN$
to another function $G:\Sigma^*\to\bbN$
is a polynomial-time deterministic Turing machine $M$
together with a polynomial 
$p(s_1,\ldots,s_t)$ such that on input word $w$, $M$ outputs $t$ strings $v_1,\ldots,v_t$ where
$F(w) = p (G(v_1),\ldots,G(v_t))$.

Note that poly-monious reductions are a slight generalization of the classical parsimonious and $c$-monious reductions, but weaker than polynomial time Turing reductions. 
Parsimonious reduction is a poly-monious reduction with the identity polynomial $p(s)=s$.
The $c$-monious reduction~\cite{bannach_p_2025} is a poly-monious reduction with the polynomial $p(s)=cs$. 
When restricted to functions in $\shpt$,
a poly-monious reduction with polynomial $p(s_1,\ldots,s_t)$ is a special case of polynomial time Turing reduction in the sense that the number of calls to the oracle is fixed to $t$, which does not depend on the input word.

\paragraph{$\shexpt$-complete functions.}
A function $F:\Sigma^* \to \bbN$ is in $\shexpt$, if there is a non-deterministic exponential time Turing machine $M$ such that for every word $w\in\Sigma^*$,
$F(w)$ is the number of accepting runs of $M$ on $w$.
It is \emph{$\shexpt$-hard}, if for every function $G \in \shexpt$, there is a poly-monious reduction from $G$ to $F$.
Finally, it is \emph{$\shexpt$-complete}, if it is in $\shexpt$ and $\shexpt$-hard.

\paragraph{Dependency Quantified Boolean Formulas (DQBF).}
A \emph{dependency quantified Boolean formula} (DQBF) in prenex normal form is a formula of the form:
\begin{align}
\label{eq:dqbf}
\Psi &\ :=\ 
\forall \vx\
\exists y_1(\vz_1) \cdots \ \exists y_k(\vz_k) 
\ \ \psi
\end{align}
where $\vx=(x_1,\ldots,x_n)$, each $\vz_i\subseteq \vx$
and $\psi$, called the \emph{matrix}, is a quantifier-free Boolean formula using variables in $\vx\cup\{y_1,\ldots,y_k\}$. 
We call $\vx$ the \emph{universal variables},
$y_1,\ldots,y_k$ the \emph{existential variables},
and each $\vz_i$ the \emph{dependency set} of $y_i$.
A $k$-DQBF is a DQBF with $k$ existential variables.
For convenience, we sometimes write 
$\exists y_i(\vz_i)$ as $\exists y_i(I_i)$ where $I_i$ is the set of indices of the variables in $\vz_i$.

A DQBF $\Psi$ as in~\eqref{eq:dqbf} is \emph{satisfiable} 
if there is a tuple $(f_1,\ldots,f_k)$, 
called \emph{Skolem functions}, such that, for every $1\leq i \leq k$, $f_i$ is a formula using only variables in $\vz_i$,
and by replacing each $y_i$ with $f_i$, the matrix $\psi$ becomes a tautology.
We call the tuple $(f_1,\ldots,f_k)$ a \emph{solution} or \emph{model} of $\Psi$
and write $(f_1,\ldots,f_k)\models\Psi$.
We refer to $\Psi$ as a \emph{uniform DQBF} if
for every model $(f_1,\ldots,f_k)\models\Psi$,
$f_1, \ldots ,f_k$ represent the same Boolean function, i.e., $|\vz_1|=\cdots = |\vz_k|=m$ and for every $\va\in\bbB^{m}$, $f_1(\va)=\cdots=f_k(\va)$.
We write $\sharp\Psi$ to denote
the number of Skolem functions of $\Psi$.

The \emph{model counting} problem for DQBF, denoted $\sharp$DQBF,
is to compute $\sharp\Psi$ for a given DQBF $\Psi$.
Its restriction to $k$-DQBF is denoted by $\sharp k$-DQBF.

\paragraph{DQBF expansion.}
We first recall the definition of the expansion of a DQBF from~\cite{dqbf-sat23}, which shows that a DQBF represents an exponentially large CNF formula.
We will need an additional notation.
For $\vz\subseteq \vx$ and $\va \in \Sigma^{|\vx|}$,
we write $\p {\va} {\vx} {\vz}$ to denote the projection of $\va$ to the components in $\vz$ according to the order of the variables in $\vx$.
For example, if $\vx=(x_1,\ldots,x_5)$ and $\vz=(x_1,x_2,x_5)$,
then $\p {\bot\bot\top\bot\top}{\vx} {\vz}$ is $\bot\bot\top$, i.e.,
the projection of $\bot\bot\top\bot\top$ to its $1^{st}$, $2^{nd}$ and $5^{th}$ bits.

Let $\Psi$ be as in Eq.~\eqref{eq:dqbf}.
For each $1\leq i\leq k$ and
for each $\vc\in \bbB^{|\vz_i|}$,
let $X_{i,\vc}$ be a variable.
For each $(\va,\vb)\in \bbB^n\times\bbB^k$,
where $\va=(a_1,\ldots,a_n)$ and $\vb=(b_1,\ldots,b_k)$,
define the clause
$C_{\va,\vb} \ := \ 
X^{\neg b_1}_{1,\vc_1}\vee \cdots \vee X^{\neg b_k}_{k,\vc_k}$,
where $\vc_i=\p {\va}{\vx}{\vz_i}$, for each $1\leq i \leq k$.
The expansion of $\Psi$, denoted by $\exp(\Psi)$, 
is the following $k$-CNF formula.
\begin{align}
\label{eq:exp}
\exp(\Psi) \ := \ &
\bigwedge_{(\va,\vb) \ \text{s.t.} \
\psi[(\vx,\vy)/ (\va,\vb)]=\bot}\quad C_{\va,\vb}
\end{align}

It is known that 
$\Psi$ is satisfiable if and only if
its expansion $\exp(\Psi)$ is satisfiable (cf.~\citeauthor{dqbf-sat23} 2023).
More precisely, a solution $(f_1,\ldots,f_k)\models\Phi$
corresponds uniquely to a satisfying assignment of $\exp(\Phi)$,
where $X_{i,\vc}=f_i(\vc)$ for every $1\leq i \leq k$ and $\vc\in\bbB^{|\vz_i|}$.

\section{Complexity of $\sharp$DQBF}\label{sec:hardness}

In this section, we will analyse the complexity of $\sharp$DQBF, starting with $\sharp{3}$-DQBF.
It is straightforward that $\sharp{3}$-DQBF is in $\shexpt$.
It is $\shexpt$-hard since every language in $\nexpt$ can be reduced parsimoniously
in polynomial time to $3$-DQBF~\cite{dqbf-sat23}.
This gives us the following theorem.

\begin{theorem}\label{theo:3dqbf}
$\sharp 3$-DQBF is $\shexpt$-complete.
\end{theorem}

\Cref{theo:3dqbf} is not surprising, given that the satisfiability problem for $3$-DQBF is already $\nexpt$-complete.
We will strengthen it by showing that $\shexpt$-hardness already holds for 2-DQBF,
whose satisfiability problem is $\ps$-complete.

Before we can prove this, we need to introduce some further terminology.
First, we recall the notion of succinct representation of graphs introduced in~\cite{avi-succ-83}.
In such a representation, instead of being given the list of edges in a graph, we are given a Boolean circuit $C(\vx,\vy)$, where $\vx,\vy$ are vectors of Boolean variables of length $n$. 
The circuit $C$ represents a graph $G_C$ where
the set of vertices is $\bbB^n$ and
there is an edge oriented from $\va$ to $\vb$, denoted $\va\to\vb$, iff
$C(\va,\vb)=\top$.

We will interpret a $2$-CNF formula $F$ as a directed graph, called the \emph{implication graph of $F$}, where each clause $(\ell_1\lor\ell_2)$ represents two edges $(\neg \ell_1\to\ell_2)$ and $(\neg\ell_2\to\ell_1)$.
If $n$ is the number of variables in $F$,
each literal can be encoded as a binary string $a_0a_1\cdots a_{\log n} \in \bbB^{1+\log n}$,
where $a_0$ is the sign 
and $a_1\cdots a_{\log n}$ is the name of the variable.

Finally, we need the notion of \emph{projection} introduced in~\cite{skyum-valiant-V85}.
Intuitively, a projection is a special kind of polynomial-time reduction where each bit $j$ in the output is determined either by the length of the input or by bit $i$ in the input, where the index $i$ can be computed efficiently from index $j$ and the length of the input.

We recall the following lemma from~\cite{dqbf-sat23}, which is inspired by the result in~\cite{yanna-succ-86}.

\begin{lemma}\label{lem:projection}\cite{dqbf-sat23}
Suppose there is a projection $\cA$ that 
takes as input a CNF formula
and outputs a graph.
Then, there is a polynomial-time algorithm 
that transforms a DQBF instance $\Psi$
to a circuit $C$ that succinctly represents 
the graph $\cA(\exp(\Psi))$.
\end{lemma}

Using Lemma~\ref{lem:projection},
we can prove the following.

\begin{lemma}\label{lem:reduction-dqbf-2dqbf}
Suppose there is a projection $\cA$ 
that takes as input a CNF formula
and outputs a 2-CNF formula.
Then, there is a polynomial time algorithm $\cB$
that transforms a DQBF instance $\Psi$ to a 2-DQBF instance $\Phi$ such that
$\sharp\Phi = \sharp \cA(\exp(\Psi))$.
\end{lemma}
\begin{proof}
Viewing $2$-CNF formula as a graph
and applying~\Cref{lem:projection},
there is a polynomial time algorithm $\cA^*$ that transforms a DQBF $\Psi$ to a circuit $C$ that succinctly represents the implication graph of $\cA(\exp(\Psi))$.

The desired algorithm $\cB$ works as follows.
Let $\Psi$ be the input DQBF.
First, run $\cA^*$ on $\Psi$ to obtain the circuit $C(u,\vx,u',\vx')$,
where $\vx,\vx'$ encode the names of variables and $u,u'$ represent the signs of literals.
Then, output the 2-DQBF
$\Phi\ :=\ \forall \vx \forall \vx'\
\exists y_1(\vx) \exists y_2(\vx')
\ \alpha\ \land\ \beta$,
where
\begin{align*}
\alpha\ := \ & (\vx=\vx')  \to  (y_1=y_2)
\\
\beta\ :=\ & \bigwedge_{b,b'\in \bbB}
C(b,\vx,b',\vx') \ \iff \ (y_1^b \to y_2^{b'})
\end{align*}
Intuitively,
$\alpha$ states that $\Phi$ is a uniform DQBF
and $\beta$ states that the implication graph of the expansion must have the same edges as $G_C$.

We claim that $\sharp \Phi=\sharp \cA(\exp(\Psi))$, i.e., 
$\sharp \Phi$ is precisely the number of solutions of the 2-CNF formula represented by the circuit $C$.
By the definition of $\beta$,
$(b,\va)\to (b',\va')$ is an edge in the graph $G_C$ iff
a clause $X_{1,\va}^b\to X_{2,\va'}^{b'}$ is in $\exp(\Phi)$.
Since $\alpha$ states that Skolem functions for $y_1,y_2$ must be the same,
the indices $1$ and $2$ in the literals 
$X_{1,\va}^b$ and $X_{2,\va'}^{b'}$ can be dropped.
It is equivalent to saying that
$(b,\va)\to (b',\va')$ is an edge in the graph $G_C$ iff
a clause $X_{1,\va}^b\to X_{1,\va'}^{b'}$ is in $\exp(\Phi)$.
Therefore, $\sharp \Phi = \sharp \cA(\exp(\Psi))$. 
\end{proof}

\begin{lemma}\label{lem:reduction-2-dqbf}
There is a polynomial-time reduction that transforms a DQBF $\Psi$ into two $2$-DQBFs $\Phi_1$ and $\Phi_2$ such that
$\sharp\Psi = \sharp\Phi_1 - \sharp\Phi_2$. 
\end{lemma}
\begin{proof}
It is shown in~\cite{bannach_p_2025} that there is a polynomial-time reduction that takes as input a CNF formula $F$ 
and outputs two $2$-CNF formulas $F_1$ and $F_2$ such that $\sharp F = \sharp F_1-\sharp F_2$.
We observe that their reduction is in fact a projection.
Using~\Cref{lem:reduction-dqbf-2dqbf}, we obtain the desired reduction.
\end{proof}

The proof of Lemma~\ref{lem:reduction-2-dqbf} is non-constructive.
We can strengthen it by giving an explicit reduction that runs in almost linear time, as stated in Lemma~\ref{lem:explicit-reduction-2-dqbf}.
The run time is quadratic in the number of existential variables and linear in the length of the matrix.

\begin{lemma}
\label{lem:explicit-reduction-2-dqbf}
There is a reduction that transforms a DQBF $\Psi$ into two $2$-DQBF $\Phi_1$ and $\Phi_2$ such that
$\sharp\Psi = \sharp\Phi_1 - \sharp\Phi_2$.
The reduction runs in time $O(k^2|\psi|)$, where $k$ is the number of existential variables in $\Psi$ and $\psi$ is the matrix of $\Psi$.
\end{lemma}

Using \Cref{lem:reduction-2-dqbf}
or \Cref{lem:explicit-reduction-2-dqbf},
we obtain the following theorem.

\begin{theorem}\label{theo:2-dqbf-shexpt-hard}
$\sharp$2-DQBF is $\shexpt$-complete.
\end{theorem}

We can also show that every $k$-DQBF can be reduced parsimoniously to a uniform $k$-DQBF, which gives us the following corollary.

\begin{corollary}\label{cor:2dqbf-shexpt-complete}
For every $k\geq 2$,
$\sharp$k-DQBF is
$\shexpt$-complete, even when restricted to uniform $k$-DQBF.
\end{corollary}

\section{First-order Model Counting (FOMC)}\label{sec:fomc}

In this section, we show a tight connection between $\sharp$DQBF and FOMC.
Recall that FOMC is defined as given FO sentence $\Psi$ and a number $N$ in unary, compute the number of models of $\Psi$ with domain $\{1,\ldots,N\}$.
We denote by $\fomcbin$ when the number $N$ is given in binary.

It is implicit in~\cite{dqbf-fmcad22} that $\fomcbin$ can be reduced to $\sharp$DQBF and that the reduction is parsimonious.
We will describe the idea here with an example.
Consider the well-known smoker-friend example for Markov Logic Networks~\cite{mln06}:
\begin{align*}
\Psi \ := \ & 
\forall u \forall v  
\ \ \stress(u)\to \smoke(u)
\\
& \qquad \land \
\friend(u,v)\land \smoke(u) \to \smoke(v)
\end{align*}
For every $n$, we will show to construct a DQBF $\Phi_n$ such that $\sharp\Phi_n$ is exactly the number of models of $\Psi$ of size $2^n$.

The idea is to represent each of the predicates $\stress$, $\smoke$, and $\friend$ with a Skolem function.
We have $2n$ universal variables $\vx_1,\vx_2$ in $\Phi_n$.
The first block of $n$ variables $\vx_1$ corresponds to $u$
and the second block $\vx_2$ corresponds to $v$.
It has $4$ existential variables, $y_1,y_2,y_3,y_4$,
corresponding to $4$ atoms $\stress(u)$,
$\smoke(u)$, $\friend(u,v)$ and $\smoke(v)$.
The dependency sets are 
$\vx_1$, $\vx_1$, $\vx_1\cup\vx_2$ and $\vx_2$, respectively.
The matrix of $\Phi_n$ is obtained by replacing each atom in $\Psi$ with its corresponding existential variable.
Formally,
\begin{align*}
\Phi_n\ :=\ & \forall \vx_1 \forall \vx_2 \ 
\exists y_1(\vx_1) \exists y_2(\vx_1)
\exists y_3(\vx_1,\vx_2) \exists y_4(\vx_2)
\ \phi,
\end{align*}
where 
\begin{align*}
\phi  :=  & \left(y_1 \to y_2\right)
\land
\left(y_3\land y_2 \to y_4\right)
\land
\left(\vx_1 = \vx_2 \to y_2=y_4\right).
\end{align*}
The first two conjuncts correspond to the quantifier-free parts in $\Psi$.
The last conjunct states that $y_2$ and $y_4$ must be the same function, since they are intended to represent the same predicate $\smoke$.
It is not difficult to show that $\sharp\Phi_n$ is precisely the number of models of $\Psi$ with size $2^n$.

Next, we show that FOMC is already hard even when the domain size is fixed to $2$.

\begin{theorem}\label{theo:fomcbin-hard}
FOMC is $\shexpt$-hard even when the domain size is fixed to $2$.
\end{theorem}
\begin{proof}
The reduction is from uniform $2$-DQBF,
which is $\shexpt$-hard, by \Cref{cor:2dqbf-shexpt-complete}.
We fix a uniform $2$-DQBF 
$\Phi:=\forall\vx \exists y_1(I)\exists y_2(J) \phi$,
where $\vx =(x_1,\ldots,x_n)$,
$I= \{i_1,\ldots,i_m\}$
and $J = \{j_1,\ldots,j_m\}$.
Let $S$ be a predicate symbol with arity $m$ and $U$ be a unary predicate.
Define the FO sentence
$\Psi := \exists u_0\exists u_1
\forall v_1 \cdots\forall v_n\
U(u_1)\land \neg U(u_0)\land \psi$,
where $\psi$ is the formula obtained from $\phi$ by replacing: (i)
each $x_i$ in $\phi$ with $U(v_i)$
for every $1\leq i\leq n$;
and (ii) $y_1$ and $y_2$ with $S(v_{i_1},\ldots,v_{i_m})$ and
$S(v_{j_1},\ldots,v_{j_m})$, respectively.
The intention is that the 
Boolean values $\top$ and $\bot$ are represented with
membership in the predicate $U$.
A Skolem function $f:\bbB^{m}\to\bbB$ is represented by the relation $S$.
We can show that $\sharp\Phi$ is half the number of models of $\Psi$ with domain $\{1,2\}$.
\end{proof}

We can show that the logic required for $\shexpt$-hardness has satisfiability problem decidable in $\ps$,
which gives us the following corollary.

\begin{corollary}\label{cor:pspace-decidable}
There is a fragment $\cL$ of FO
of which the satisfiability problem is in $\ps$, but its corresponding FOMC is $\shexpt$-complete even when the domain size is restricted to $2$.
\end{corollary}

\section{Algorithm for $\shtwodqbf$}\label{sec:compute}

In this section, we present an algorithm for $\sharp 2$-DQBF that builds on recent advances in 2-DQBF satisfiability checking using symbolic reachability~\cite{dqbf-fmcad24}.
The key idea is to interpret the matrix of a $2$-DQBF as a succinct encoding of the implication graph induced by its expansion. Our algorithm symbolically decomposes this graph into its weakly connected components and computes the model count by processing each component independently.

We fix the input $2$-DQBF 
$\Phi  \coloneq  \forall \vx \exists y_1(\vz_1) \exists y_2(\vz_2)\; \varphi$.
Instead of computing $\sharp\Phi$ directly, 
we will compute $\sharp\exp(\Phi)$.
There is a difference because 
a variable $X_{i,\vc}$ may not even occur 
in $\exp(\Phi)$, indicating that the Skolem function of $y_i$ is completely 
unconstrained at assignment~$\vc$.
We call a variable $X_{i,\vc}$ a \emph{support} variable if it appears in $\exp(\Phi)$; otherwise, it is called \emph{non-support}.
Since non-support variables can be assigned arbitrarily,
it is sufficient to compute the number of solutions that assign non-support variables to a fixed value, say, $\bot$.
We call such solutions \emph{essential solutions}.

\paragraph{Counting non-support variables.} $\sharp\Phi$ can be recovered from the number of essential solutions by multiplying it with $2^{m}$, where $m$ is the number of non-support variables.
The set of support/non-support variables can be characterised with Boolean formulas as follows.
\begin{lemma}\label{lem:support-var}
Let $\cS_1 \coloneq \{\vc : \neg \varphi[\vz_1/\vc] \ \text{is satisfiable}\}$
and $\cS_2 \coloneq \{\vc : \neg \varphi[\vz_2/\vc] \ \text{is satisfiable}\}$.
The set of support variables in $\exp(\Phi)$
is $\{X_{1,\vc} : \vc\in \cS_1\}\cup \{X_{2,\vc} : \vc\in \cS_2\}$.
Moreover, the number of support and non-support variables is $|\cS_1|+|\cS_2|$ and  $(2^{|\vz_1|}-|\cS_1|)+(2^{|\vz_2|}-|\cS_2|)$, respectively.
\end{lemma}

Given a BDD for the negated matrix~$\neg \varphi$, Lemma~\ref{lem:support-var} can be used to efficiently compute the number of support variables $|\cS_i|$ by projecting out variables not in $\vz_i$ and counting satisfying assignments.

\paragraph{Overview of the algorithm.}
In the following,
let $\gphi$ be the implication graph of $\exp(\Phi)$.
Let $\tgphi$ be the undirected graph obtained from $\gphi$ by ignoring the edge orientation and
adding an edge between a literal and its negation, for every literal in $\exp(\Phi)$.
The connected components of $\tgphi$ correspond to a partition of the clauses in $\exp(\Phi)$ where no two components share common variables.

High-level pseudocode is shown as~\Cref{algo:main}.
First, using the reduction in~\cite{dqbf-fmcad24},
we convert $\Phi$ to a transition system $(I,T)$,
where $I$ is the formula for the initial states and 
$T$ is the formula for the transition relation.
A brief summary of this transformation can be found in the appendix.
From $(I,T)$, we can deduce whether $\Phi$ is satisfiable by constructing a formula $\varphi_{tr}$ that represents the transitive closure of $\gphi$ via BDD-based reachability.
If it is not satisfiable, 
the algorithm immediately returns $0$.

Now, suppose $\Phi$ is satisfiable.
From $\varphi_{tr}$, we can also construct the formula for $\tgphi$.
\Cref{algo:main} iterates through every connected component $C$ in $\tgphi$ that contains only the support variables.
In each iteration, it computes $N_C$,
the number of assignments on the variables in $C$ that respect the implications in $C$.
For example, if there is an edge $\ell_1\to\ell_2$ in $\gphi$,
when $\ell_1$ is assigned to $\top$,
$\ell_2$ must also be assigned to $\top$.
If there are $k$ connected components $C_1,C_2,\ldots,C_k$ (that contains only support variables),
then the number of essential solutions is the product $\prod_{1\leq i \leq k} N_{C_i}$,
since no two components share the same variable.

\begin{algorithm}[t]\caption{Count the number of essential solutions for $\Phi$}
\label{algo:main}{\footnotesize
\begin{algorithmic}[1]
\State
Transform $\Phi$ to a symbolic reachability instance $(I,T)$ using the transformation in ~\cite{dqbf-fmcad24}
\If {$\Phi$ is unsatisfiable}
\State \textbf{return} 0
\EndIf
\State
$R\gets$ the set of all support variables
\State
$N \gets 1$
\While{$R\neq \emptyset$}
\State
Pick an arbitrary variable $X_{i,\vc}$ from $R$
\State
$C\gets$ the connected component in $\tgphi$ that contain $X_{i,\vc}$
\State
$N_C \gets$ the number of assignments on the variables in $C$
\Statex
\hspace{2cm}that respect the implications in~$C$
\State
Remove all the variables in $C$ from $R$
\State
$N \gets N \times N_C$
\EndWhile
\State
\textbf{return} $N$
\end{algorithmic}}
\end{algorithm}

\paragraph{Counting over a component.}
This is the most technical part of the algorithm.
We start with the following lemma on efficient model counting for $1$-DQBF.

\begin{lemma}\label{lem:count-one-dqbf-all}
Let $\Upsilon\coloneq\forall \vu \exists y(\vv) \ \varphi$
be a satisfiable 1-DQBF.
\begin{itemize}
\item
The number of Skolem functions for $\Upsilon$ is $2^{m}$, where $m=2^{\abs{\vv}} - |\{\vc : \neg\varphi[\vv/\vc] \ \text{is satisfiable}\}|$. 
\item 
In particular, for a set $S \subseteq \bbB^{\abs{\vv}}$, the number of Skolem functions for $\Upsilon$ that differ on $S$ is $2^m$, where $m = {\abs{S}} - |\{\vc : \neg\varphi[\vv/\vc] \land (\vc\in S) \ \text{is satisfiable}\}|$.
\end{itemize}
\end{lemma}

\Cref{lem:count-one-dqbf-all} tells us
that if we have a candidate Skolem function $f$ for $y_1$,
by substituting $y_1$ with $f$, we obtain a $1$-DQBF instance $\Phi'$
of which the number of Skolem functions restricted to a component can be computed efficiently
via \Cref{lem:count-one-dqbf-all}.
We perform this for every candidate Skolem function for $y_1$
to compute the number $N_C$.

The main challenge is to enumerate all possible Skolem functions for $y_1$.
To do so, we combine the candidate Skolem function enumeration technique in~\cite{pedant21}
and the Skolem function extraction in~\cite{dqbf-fmcad24}.
Suppose we already have a list of Skolem functions $F = \{f^{(1)},\ldots,f^{(t)}\}$ for $y_1$, where each
$f^{(i)}$ is given as a Boolean formula.
To find a Skolem function different from all functions in this list, it must differ from each $f^{(i)}$ at some $\vv_i$. Let $A$ be a Boolean formula, over variables $\vv_1,\ldots,\vv_t$ where each $|\vv_i|=|\vz_1|$, maintained throughout the enumeration process.
We will use $A$ to represent the assignments we can choose to differ from each $f^{(i)}$. 
The intuition is that if $\ttM$ is a satisfying assignment for $A$, then we want to find a function $f$ with $f(\ttM(\vv_i)) = \neg f^{(i)}(\ttM(\vv_i))$ for every
$1\leq i \leq t$.

How can we construct the Boolean formula that defines the function $f$?
Here we employ the technique from~\cite{dqbf-fmcad24}.
First, we ``force'' the variable $X_{1,\ttM(\vv_i)}$ to be assigned with $\neg f^{(i)}(\ttM(\vv_i))$
by adding the edge 
\[
X_{1,\ttM(\vv_i)}^{f^{(i)}(\ttM(\vv_i))} \to
X_{1,\ttM(\vv_i)}^{\neg f^{(i)}(\ttM(\vv_i))}
\] 
into the transition relation $T$, for every $1\leq i \leq t$.
We then check whether in the transition relation $T$ there is a cycle that contains contradicting literals.
If there is such a cycle, we move to the next satisfying assignment of $A$ by ``blocking'' the assignment $\ttM$ in $A$.
If there is no such cycle,
we extract the function $f$ for $y_1$ by employing the technique from~\cite{dqbf-fmcad24},
add $f$ into $F$, and update $A$ by conjoining it with $C_1[\vz_1 / \vv_{t+1}]$, where $\vv_{t+1}$ are fresh variables, and $C_1$ is a Boolean formula extracted from $C$ that specifies only the literals associated with $y_1$.

Without additional constraints, we would enumerate many satisfying assignments $\ttM$ of $A$ which do not lead to a Skolem function. For instance, if $\ttM(\vv_i) = \ttM(\vv_j)= \va$, but $f_i(\va) \neq f_j(\va)$, there clearly is no function $f_{t+1}$ such that both $f_{t+1}(\va) \neq f_i(\va)$ and $f_{t+1}(\va) \neq f_j(\va)$.
Such cases, and many more, can be excluded by conjoining~$A$ with{\footnotesize
\begin{align}
\label{eq:update-A}
 \
\bigwedge_{1\leq i \leq t} 
\neg \varphi_{tr}\left(X^{\neg f^{(t+1)}(\vv_{t+1})}_{1, \vv_{t+1}}, X^{f^{(j)}(\vv_{j})}_{1, \vv_{j}} \right).
\end{align}}\!\!
Recall that $\varphi_{tr}$ is a formula that represents the edges of the transitive closure of the implication graph. The formula $\varphi_{tr}\left(X^{\neg f^{(i)}(\vz^{(i)})}_{1, \vz^{(i)}}, X^{f^{(j)}(\vz^{(j)})}_{1, \vz^{(j)}}\right)$ represents the formula obtained by substituting the variable representing the two literals in $\varphi_{tr}$ with $X^{\neg f^{(i)}(\vz^{(i)})}_{1, \vz^{(i)}}, X^{f^{(j)}(\vz^{(j)})}_{1, \vz^{(j)}}$.

The intended meaning of Eq.~\eqref{eq:update-A} is as follows.
The conjunct $C_1[\vz_1 / \vv_{t+1}]$ states that
the place $\vv$ where the next Skolem function differs from $f$ must be in component $C_1$.
Each conjunct $\neg \varphi_{tr}\left(X^{\neg f^{(i)}(\vz^{(i)})}_{1, \vz^{(i)}}, X^{f^{(j)}(\vz^{(j)})}_{1, \vz^{(j)}}\right)$ ensures that the edge 
$X^{\neg f^{(i)}(\vz^{(i)})}_{1, \vz^{(i)}}\to X^{ f^{(j)}(\vz^{(j)})}_{1, \vz^{(j)}}$ is \emph{not}
in the transitive closure of $G_{\Phi}$.
Otherwise, if such an edge is present,
when we force $X_{1, \ttM(\vv_i)}$ to be $\neg f^{(i)}(\ttM(\vv_i))$ and $X_{1, \ttM(\vv_j)}$ to be $\neg f^{(j)}(\ttM(\vv_j))$, we will find a bad cycle and there will be no solutions.

\begin{algorithm}[t]
\caption{Counting $N_C$}\label{algo:count-on-component}
{\footnotesize
\begin{algorithmic}[1]
\State 
Let $C_1,C_2$ be the literals in $C$ associated with $y_1,y_2$, resp.
\State $F,A,n \gets \emptyset,\top,0$
\State $T'\gets T$\Comment{$T$ is the transition relation constructed from $\Phi$}
\While{$A$ is satisfiable}
\State Let $\texttt{M}$ be a satisfying assignment of $A$
        \LineComment{Force the candidate to be different from the previous ones}
            \State $T' \gets T' \land \Call{ForceAssignment}{\texttt{M},F}$
        \LineComment{Check if such assignments lead to no Skolem functions}
        \State $E' \gets \Call{ComputeReachable}{E, T'}$
        \If {$\Call{CheckBadCycle}{E'}$}
            \State $A \gets A \land \Call{BlockAssignment}{\texttt{M}}$
            \State \textbf{continue}
        \EndIf
        \LineComment{Count the number of Skolem functions and update $A$}
        \State $f \gets \Call{ComputeValidCandidate}{T'}$
        \State $n \gets n + \Call{Count1DQBFOnComponent}{f}$
        \State $F \gets F \cup \{f\}$
        \State $\Call{Update}{A}$
    \EndWhile
    \State \textbf{return} $n$\Comment{$n$ is the number $N_C$}
\end{algorithmic}}
\end{algorithm}

We present the algorithm to compute $N_C$
formally as Algorithm~2.
We first split $C$ into two sets $C_1$ and $C_2$
that contain the literals associated with $y_1$ and $y_2$, respectively.

In Line~8, we force the assignment by adding into $T$ the edge 
$X_{1,\ttM(\vv_i)}^{f^{(i)}(\ttM(\vv_i))} \to
X_{1,\ttM(\vv_i)}^{\neg f^{(i)}(\ttM(\vv_i))}$, for every $1\leq i \leq t$.
In Line~9 we run the command \tool{reach} from \tool{ABC}
to check whether there is a cycle containing contradicting literals.
If there is such a cycle, the assignment $\ttM$ is blocked in Line~11.
In Line~13 we extract a Skolem function using the technique from~\cite{dqbf-fmcad24}.
In Line~14, we count the number of solutions for the 1-DQBF instance after substituting $y_1$ with the new Skolem function $f$.
Finally, in Line~16 we update the formula $A$ by conjoining it with the conjunction in Eq.~\eqref{eq:update-A}.

\section{Experiments}\label{sec:experiment}

\begin{figure*}[t!]
\centering
\includegraphics[scale=.31]{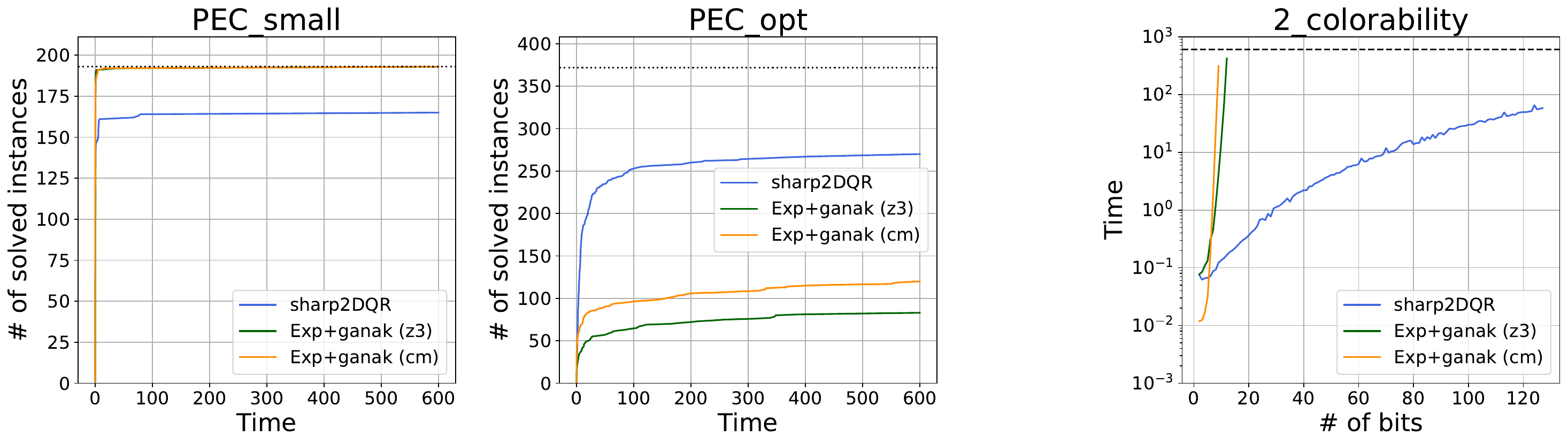}%
\caption{The two figures on the left show the performance of the solvers on the PEC instances, the horizontal axis corresponds to the running time (s), and the vertical
axis to the number of solved instances. The figure on the right shows the performance of the solvers on the 2-colorability instances; the horizontal axis corresponds to the number of bits of the graph in the instance, and the vertical axis to the running time (s).}
\label{fig:pec}
\end{figure*}

We implemented the algorithm from the previous section in a tool called \tool{sharp2DQR}.
Formulas such as $R$, $G_{\Phi}$, $\varphi_{tr}$, $\cS_1$ and $\cS_2$ are represented as BDDs using \tool{cudd}~\cite{somenzi2009cudd}.
This offers several advantages.
For example, the formula
$\varphi_{tr}$
can be computed using BDD-based reachability as implemented in ABC's \texttt{reach} command~\cite{ABC}. 
The number of support/non-support variables can also be computed easily
by constructing the BDD for $\cS_1$ and $\cS_2$ from $\neg\varphi$ with existential quantification.

To evaluate our model counter, we generated a diverse family of benchmarks, which we divided into three batches of instances.
\begin{itemize}
\item 
PEC\_opt: These are instances with a dependency set size of 10 to 50. They are generated in a similar manner as in~\cite{dqbf-fmcad24}.
There are 370 instances in this batch.
One third of these have $0$ non-support variables.
\item 
PEC\_small: These are instances generated from the ISCAS89 instances with a dependency set size of 3 to 10 variables. 
There are 192 instances in this batch.
\item 
2\_colorability: These are the 2 colorability instances as in~\cite{dqbf-fmcad24}. 
These are succinctly represented graphs with 2 to 127 bits, each of them contains exactly two Skolem functions.
\end{itemize}
For PEC\_opt and PEC\_small, the number of Skolem functions ranges from 1 to more than $2^{2^{64}}$
and the number of connected components ranges from 1 to more than 1600.

We evaluated the performance of \tool{sharp2DQR}
against \tool{ganak}~\cite{ganak}, which is used to count the number of models of the expansion as defined in Eq.~\eqref{eq:exp}. The expansion is computed via \tool{cryptominisat5}~\cite{cryptominisat}  or \tool{z3}~\cite{z3} as follows: First, we obtain a model $\texttt{M}$ of $\neg\varphi$. Then, we generate the corresponding clauses in the expansion. Then, we add a blocking clause $\land_{x \in \vz_1 \cup \vz_2 \cup \{y_1, y_2\}} x \neq \texttt{M}[x]$ to $\neg\varphi$ to prevent duplicated solutions and repeat until the formula becomes unsatisfiable. This method is called \tool{Exp+ganak}.

The experiments were conducted on Ubuntu 22.04.4 LTS with 48 GB of 2400MHz DDR4 memory and an i5-13400 CPU.
Each solver had 600 seconds to solve each instance. 

\Cref{fig:pec} shows that
\tool{sharp2DQR} fell short on PEC\_small instances, but it is significantly better than \tool{Exp+ganak} on PEC\_opt instances. This is because the dependency set size is larger on PEC\_opt instances, and since the expansion size is exponential to the dependency set size, our method, without expanding, performs better on larger instances. 
Most of the time spent by \tool{Exp+ganak} is used on computing the expansion, and on many instances \tool{ganak} finishes counting quite quickly after the expansion.
\tool{sharp2DQR} does not work well on small instances because sometimes BDD operations take too long, and on the PEC\_small instances, the number of unsatisfying models is small enough that enumeration is not too much of a problem.
We also notice that for \tool{Exp+ganak}, the performance of using \tool{cryptominisat5} and \tool{z3} is similar, but \tool{cryptominisat5} is better on the PEC\_opt instances.

For 2\_colorability, \tool{Exp+ganak} was unable to solve any instances larger than $12$ bits, while \tool{sharp2DQR} successfully solved instances up to $127$ bits. 
This is due to the fact that the number of clauses in the expansion is $\Theta(2^n)$ for an $n$-bit graph, making full expansion very expensive.

More experimental results and analysis can be found in the appendix.
We also compared both \tool{Exp+ganak}
and \tool{sharp2DQR} against the latest FOMC tool \tool{WFOMC}~\cite{wfomc_tool}, 
where we encode counting the number of 2-colorings and independent sets
on some specific graphs.
In all instances, \tool{WFOMC} can only handle model sizes of up to $4$, far lower than what \tool{sharp2DQR} can handle, which in some instances is $2^{127}$.
However, in the independent set counting instances, \tool{Exp+ganak} performs better than \tool{sharp2DQR}.

\section{Concluding Remarks}\label{sec:conclusion}

We established that $\shtwodqbf$ is as hard as general $\sharp$DQBF. 
Specifically, we proved that it
is $\shexpt$-complete by leveraging the connections between $k$-DQBF and $k$-SAT~\cite{dqbf-sat23} and the technique in~\cite{bannach_p_2025}.

On the experimental front, we introduced a novel algorithm for $\shtwodqbf$ using BDD-based symbolic reachability.
As a baseline, we also implemented an approach that relies on universal expansion followed by propositional model counting. While our algorithm scaled better with larger dependency sets, the expansion-based method works for general DQBF and may be worth exploring further.

To the best of our knowledge, this is the first paper investigating model counting for DQBF, and there are many avenues for future research.
One natural next step is to generalize our algorithm to handle 3-DQBF.

\subsection{Acknowledgements}
This paper is the extended version of the conference paper with same title~\cite{us-aaai-26}.
We thank the AAAI'26 reviewers for their detailed and constructive comments on the initial version.

We would like to acknowledge the generous support of the Royal Society International Exchange Grant
no.\@ R3\textbackslash{}233183,
the National Science and Technology Council of Taiwan grant no.\@ 111-2923-E-002-013-MY3 and 114-2221-E-002-183-MY3,
and the NTU Center of Data Intelligence: Technologies, Applications, and Systems 
grant no.\@ NTU-113L900903.

\bibliography{references}

\appendix
\onecolumn

\section*{Appendix: More detailed proofs and experimental results\\
``Model Counting for Dependency Quantified Boolean Formulas''}

\noindent
\underline{\hspace{\textwidth}}

\vspace{-.35cm}
\noindent
\underline{\hspace{\textwidth}}

\bigskip
\medskip
\noindent
In this appendix we present some of the missing proofs in the main body.
We also present more experimental results and the statistics.
All logarithms have base $2$.

\section{Parsimonious reduction from general DQBF to uniform DQBF}\label{app:sec:dqbf-uniform-dqbf}

Recall that a $k$-DQBF $\Phi$ is uniform,
if for every Skolem functions $(f_1,\ldots,f_k)\models \Phi$,
all the functions $f_1,\ldots,f_k$ must be the same function.
Here we will show the polynomial time parsimonious reduction from general DQBF to uniform DQBF.
The intuition is that we can combine 
$k$ functions $f_1,\ldots,f_k$ into one function $f$.
Suppose $f_i : \bbB^{n_i}\to \bbB$, for each $1\leq i \leq k$.
We can combine it into one function $f:\bbB^{\log k}\times \bbB^{n}\to \bbB$,
where $n$ is the maximal among $n_1,\ldots,n_k$.
The first $\log k$ bits can be used as the index of the function.

We now give the formal construction.
We fix an arbitrary $k$-DQBF:
\begin{align*}
\Psi &\ :=\ 
\forall \vx\
\exists y_1(\vz_1) \cdots \ \exists y_k(\vz_k) 
\ \ \psi
\end{align*}
where $\vx=(x_1,\ldots,x_n)$, each $\vz_i\subseteq \vx$.

We construct the following uniform $k$-DQBF:
\[
\Psi'=\forall \vx_1,\dots,\forall \vx_k,\forall \vu_1,\dots,\forall \vu_k,\exists y_1(\vx_1\cup \vu_1),\dots,\exists y_k(\vx_k\cup \vu_k).\,\psi'\,,
\]
where each $\vx_i$ is a copy of $\vx$,
each $\vu_i$ is of $\lceil\log k\rceil$ bits encoding a number in $[2^{\lceil\log k\rceil}]$,
and $\psi'$ is the conjunction of the following:
\begin{itemize}
\item 
$((\vx_i=\vx_j)\wedge (\vu_i=\vu_j))\to (y_i\leftrightarrow y_j)$ for every $1\leq i<j\leq k$,
\item 
$((\vx_1|_{\vz_i}=\vx_2|_{\vz_i})\wedge (\vu_1=\vu_2=i))\to (y_1\leftrightarrow y_2)$ for every $i\in[k]$,
\item 
$\Paren*{\bigwedge_{i\in[k]}\vu_i=i}\to \psi$, and
\item 
$(\vu_i=j)\to y_i$ for every $i\in[k]$ and $k+1\leq j\leq 2^{\lceil\log k\rceil}$.
\end{itemize}
The first item ensures that the DQBF is uniform.
The second ensures that the dependencies of the original existential variables are respected.
The third encodes the original formula.
The last enforces that there is a unique Skolem function for the unused indices in order to preserve the number of Skolem functions.

For every Skolem function $(f_1, \dots, f_k)$, consider the function: 
\[
f(i, \vx) = \begin{cases}
f_i(\vx\vert_{\vz}) & \text{ if } i \leq k\\
1 & \text{otherwise}
\end{cases}
\]
which forms a Skolem function for $\Psi'$ by copying itself $k$ times. 
For any Skolem function of $\Psi'$, it must be $k$ copies of a function $f$.
Consider the function:
\[
f_i(\vz_i) = f(i, \text{Ext}(\vz_i, \vx))
\]
where $\text{Ext}(\vz_i, \vx)$ lifts $\vz_i$ to $\vx$ by substituting the variables not in $\vz_i$ to $0$. Then $(f_1, \dots, f_k)$ is a Skolem function for $\Psi$. It is clear that this is a bijection.

\section{Proof of~\Cref{lem:explicit-reduction-2-dqbf}: An explicit reduction from $\sharp$DQBF to $\sharp$2-DQBF}

In this section we present the proof of~\Cref{lem:explicit-reduction-2-dqbf}:
\begin{quote}{\em
There is a reduction that transforms a DQBF $\Psi$ to two $2$-DQBF $\Phi_1$ and $\Phi_2$ such that
$\sharp\Psi = \sharp\Phi_1 - \sharp\Phi_2$.
The reduction runs in time $O(k^2|\psi|)$, where $k$ is the number of existential variables in $\Psi$ and $\psi$ is the matrix of $\Psi$.}
\end{quote}
We will first briefly review the reduction in~\cite{bannach_p_2025}.
For the sake of presentation, we will present only the simplified version, which is already sufficient for our purpose.
It should also be clear from our presentation that the reduction is indeed projection as required in the proof of~\Cref{lem:reduction-2-dqbf}.
We then show how to lift the reduction to the DQBF setting.
We will first show the reduction to a slightly generalized version of 2-DQBF that we call \emph{extended 2-DQBF}.
Then, we present a parsimonious reduction from extended 2-DQBF to 2-DQBF.

\subsection{The reduction from $\sharp$SAT to $\sharp$2-SAT:}

Given a CNF $\varphi$ over the variables $\vx=(x_1,\ldots,x_n)$: 
\[
\varphi := \bigwedge_{i=1}^m C_i,\quad\text{where each}\ C_i = \ell_{i, 1} \lor \dots \lor \ell_{i,k_i}
\]
We construct a 2-CNF forula $\phi$ over the variables $\vx\cup\{c_1,\ldots,c_m\}\cup\{p_i, o1_i, o2_i, e1_i, e2_i\}_{i=0,\ldots,m}$.
The intended meaning of each variable is as follows.
\begin{itemize}
\item 
Variable $c_i$ is the indicator whether clause $C_i$ is falsified:
It is true only if clause $C_i$ is falsified.
\item 
Variable $p_i$ denotes the parity of the number of clauses falsified up to clause $i$.
\item 
Variables $o1_i$, $o2_i$, $e1_i$, and $e2_i$ are used to represent different cases of the parity of the number of clauses falsified up to clause $i$. 
\end{itemize}
The formula $\phi$ is constructed to reflect the intended meaning of the variables.
Formally, it is the conjunction of the following clauses:
\begin{description}
\item[(R1)]
For each $i = 1,\dots,m$, and every literal $\ell_{i, j}$, we have the clause: 
\begin{equation}
c_i \to \neg \ell_{i, j} \label{eq:clause}
\end{equation}
That is, if clause $C_i$ is falsified, then all its literals must be falsified.
\item[(R2)]
For each $i = 1,\dots,m$, we have the following four groups of clauses:
\begin{itemize}
\item 
The number of clauses falsified up to clause $i$ is odd because clause $i$ is falsified.
\begin{equation}
o1_i \to \neg p_{i - 1} \quad\quad o1_i \to c_i \quad\quad o1_i \to p_i \label{eq:odd1}
\end{equation}
\item 
The number of clauses falsified up to clause $i$ is odd because clause $i$ is not falsified. 
\begin{equation}
o2_i \to p_{i - 1} \quad\quad o2_i \to \neg c_i \quad\quad o2_i \to p_i \label{eq:odd2}
\end{equation}
\item 
The number of clauses falsified up to clause $i$ is even because clause $i$ is falsified.
\begin{equation}
e1_i \to p_{i - 1} \quad\quad e1_i \to c_i \quad\quad e1_i \to \neg p_i \label{eq:even1} 
\end{equation}
\item 
The number of clauses falsified up to clause $i$ is even because clause $i$ is not falsified.
\begin{equation}
e2_i \to \neg p_{i - 1} \quad\quad e2_i \to \neg c_i \quad\quad e2_i \to \neg p_i \label{eq:even2}
\end{equation}
\end{itemize}
\item[(R3)]
Finally we have the clause:
\begin{equation}
\neg p_0 \label{eq:initial}
\end{equation}
Intuitively, this means that 
initially there is zero clause falsified.
\end{description}
In~\cite{bannach_p_2025}, it is shown that:
$$
\#(\varphi) \ =\ \#(\phi \land \neg p_m) - \#(\phi \land p_m)
$$
The proof is essentially the inclusion-exclusion principle, where we count the number of falsifying assignments.

\subsection{The reduction from $\sharp$DQBF to $\sharp2$-DQBF:}

We will show how to reduce DQBF to \emph{extended 2-DQBF} defined as follows.

\begin{definition}
A DQBF is an \emph{extended 2-DQBF} if its matrix is a conjunction $\bigwedge_i \varphi_i$, where each $\varphi_i$ uses at most two existential variables.
\end{definition}

In the following we will show the desired reduction from DQBF to extended 2-DQBF which suffices for our purpose thanks to the following lemma.

\begin{lemma}
There is a parsimonious polynomial time reduction that transforms extended 2-DQBF instances to 2-DQBF instances.
\end{lemma}
\begin{proof}
The proof is rather similar to~\Cref{app:sec:dqbf-uniform-dqbf}, where
we combine several functions into one function.

Let the given extended $2$-DQBF be: 
\[
\forall \bar{x} \exists y_1(\bar{z}_1) \ldots \exists y_k(\bar{z}_k) \; \varphi,
\quad\text{where}
 \
\varphi := \bigwedge_{1 \leq i < j \leq m} \varphi(\bar{x}, \bar{y}_i, \bar{y}_j).
\]
Consider the 2-DQBF:
\[
\forall \bar{x} \forall \bar{x}' \forall \bar{i} \forall \bar{i}' \exists y(\bar{x}, \bar{i}) \exists y'(\bar{x}', \bar{i}'). \; \psi
\]
where 
\begin{align*}
\psi' = &\bigwedge_{k = 1}^m \left( (\bar{i} = \bar{i}' = k \land \bar{z}_k = \bar{z}_k') \to y = y'\right) \ \land\ \left(\bar{x} \neq \bar{x}' \lor \bigwedge_{1 \leq k < \ell \leq m}\left((\bar{i} = k \land \bar{i}' = \ell) \to \varphi_{k, \ell}[y_k / y, y_\ell/y']\right)\right)
\end{align*}

Intuitively, the function $y$ represents $y_k$ when $i = k$. When $i = i' = k$ and $k$ is between $1$ and $m$, the formula $(\bar{i} = \bar{i}' = k \land \bar{z}_k = \bar{z}_k') \to y = y'$ forces $y$ and $y'$ to be the same function, and is independent of variables outside $\bar{z}_k$. Since each $\varphi_{k, \ell}$ only contains two existential variables, $y_k$ and $y_\ell$, $(\bar{i} = k \land \bar{i}' = \ell) \to \varphi_{k, \ell}[y_k / y, y_\ell/y']$ ensures that $y_k$ and $y_\ell$ represented by $y$ satisfies the $\varphi_{k, \ell}$ when $\bar{x}= \bar{x}'$.

\end{proof}

We now give the reduction as required in \Cref{lem:explicit-reduction-2-dqbf}.
Let the input $k$-DQBF be:
\[
\Phi := \forall x_1 \cdots x_n \exists y_1(\bar{z}_1) \cdots y_k(\bar{z}_k) \;\varphi.\]
Let $\Lambda = \{0, 1\}^{n} \times \{0, 1\}^{k}$ and $\Lambda^* = \Lambda \cup \{-1\}$.
We will view elements in $\Lambda$ as integers between $0$ and $2^{n + k}-1$ (inclusive), and elements in $\Lambda^*$ as integers between $-1$ and $2^{n + k}-1$ (inclusive).
By encoding $-1$ with an additional bit and abuse of notation, we may treat $\Lambda^*$ as a subset of $\bbB^{n+k+1}$.

Consider the following DQBF:
\begin{align*}
\Psi :=\ &
\forall (\bar{x}, \bar{y}) \in \Lambda \; \forall \vt \in \Lambda^* \; \forall \vt' \in \Lambda^* \ \exists f_1(\bar{z}_1) \cdots \exists f_k(\bar{z}_k) \
\exists c(\bar{x}, \bar{y}) \
\exists p(\vt') \
\exists o1(\vt) \; \exists o2(\vt)\
\exists e1(\vt) \; \exists e2(\vt)\ \ \psi
\end{align*}
where $|\vx|=n$, $|\vy|=k$, $|\vt|=|\vt'|=n+k+1$
and the matrix $\psi$ is the conjunction of the following

\begin{description}
    \item[(S1)] \hfill

    \begin{itemize}
        \item 
        $\neg \varphi \to (y_i \to (c \to f_i))$ for each $i \in \{1, \dots, k\}$
        \item 
        $\neg \varphi \to (\neg y_i \to (c \to \neg f_i))$ for each $i \in \{1, \dots, k\}$ 
        \item 
        $\varphi \to \neg c$
    \end{itemize}

    \item[(S2)] \hfill

    \begin{itemize}
        \item 
        $(t \neq -1 \land t' = t - 1) \to (o1 \to \neg p)$
        \item 
        $(t \neq -1 \land t = (\bar{x}, \bar{y})) \to (o1 \to c)$
        \item 
        $(t \neq -1 \land t' = t) \to (o1 \to p)$
        \\
        \item 
        $(t \neq -1 \land t' = t - 1) \to (o2 \to p)$
        \item 
        $(t \neq -1 \land t = (\bar{x}, \bar{y})) \to (o2 \to \neg c)$
        \item 
        $(t \neq -1 \land t' = t) \to (o2 \to p)$
        \\
        \item 
        $(t \neq -1 \land t' = t - 1) \to (e1 \to p)$
        \item 
        $(t \neq -1 \land t = (\bar{x}, \bar{y})) \to (e1 \to c)$
        \item 
        $(t \neq -1 \land t' = t) \to (e1 \to \neg p)$
        \\
        \item 
        $(t \neq -1 \land t' = t - 1) \to (e2 \to \neg p)$
        \item 
        $(t \neq -1 \land t = (\bar{x}, \bar{y})) \to (e2 \to \neg c)$
        \item 
        $(t \neq -1 \land t' = t) \to (e2 \to \neg p)$
    \end{itemize}

    \item[(S3)] \hfill

    \begin{itemize}
        \item $(t' = -1) \to \neg p$
    \end{itemize}
\end{description}
The formulas in (S1) correspond to the clauses in (R1), (S2) to the clauses in (R2) and (S3) to the clause in (R3). 

Consider the following two 2-DQBF $\Psi_1$ and $\Psi_2$.
\begin{itemize}
\item 
$\Psi_1$ is obtained from $\Psi$ by conjuncting its matrix with 
$(t' = 2^{n + 3} - 1) \to \neg p$.
\item 
$\Psi_2$ is obtained from $\Psi$ by conjuncting its matrix with 
$(t' = 2^{n + 3} - 1) \to p$.
\end{itemize}
It can be shown that:
$$
\sharp\Phi = \sharp\Psi_1 - \sharp\Psi_2.
$$
The proof is similar to ~\cite{bannach_p_2025}. We observe that the expansion of $\Psi_1$ and $\Psi_2$ are essentially the same as in the previous subsection.
Note also that both $\Psi_1$ and $\Psi_2$ are extended 2-DQBF.

\section{Missing details in the proof of \Cref{theo:fomcbin-hard}}

We first recall the reduction in~\Cref{theo:fomcbin-hard}.
Suppose we are given a uniform $2$-DQBF: 
\[
\Phi:=\forall\vx \exists y_1(I_1)\exists y_2(I_2) \phi,
\]
where $\vx =(x_1,\ldots,x_n)$,
$I= \{i_1,\ldots,i_m\}$
and $J = \{j_1,\ldots,j_m\}$.
We construct the following FO sentence $\Psi$ using only one predicate symbol $S$ with arity $m$:
$$
\Psi\ :=\ \exists u_0\exists u_1
\forall v_1 \cdots\forall v_n\
(u_0\neq u_1)\land \psi
$$
where $\psi$ is the formula obtained from $\phi$ by replacing:
\begin{itemize}
\item
each $x_i$ in $\phi$ with $v_i=u_1$
for every $1\leq i\leq n$;
\item
$y_1$ and $y_2$ with $S(v_{i_1},\ldots,v_{i_m})$ and
$S(v_{j_1},\ldots,v_{j_m})$, resp.
\end{itemize}
The rest of this appendix is devoted to the proof that $\sharp\Phi$ is half the number of models of $\Psi$ with domain $\{1,2\}$.
To avoid being repetitive, in this section, the domain of first-order structures is assumed to be $\{1,2\}$.

We first show that swapping the roles of $1$ and $2$ in a structure does not effect the satisfiability of $\Psi$.
Suppose $\cA\models \Psi$.
Let $\cA^*$ be the structure obtained by swapping the roles of $1$ and $2$ in $\cA$.
We claim that $\cA\models \Psi$ iff $\cA^*\models \Psi$.
Indeed, if $\cA\models \Psi$,
by definition, there is an assignment to $u_0,u_1$ with the elements in $\{1,2\}$.
Due to $u_0\neq u_1$, the assignments must be different.
Suppose $u_0$ is assigned with $1$ and $u_2$ is assigned with $2$.
Then, $\cA^*$ also satisfies $\Psi$
by assigning $u_0$ with $2$ and $u_1$ with $1$.
That $\cA^*\models \Psi$ implies
$\cA\models \Psi$ is analogous.

Next, we show that
a Boolean function $f:\bbB^m\to\bbB$ corresponds uniquely to two structures $\cA_{1,f}$ and $\cA_{2,f}$.
We use the following notation.
For each $\va =(a_1,\ldots,a_m)\in \bbB^m$, define $\widetilde{\va}=(\widetilde{a}_1,\ldots,\widetilde{a}_m)\in\{1,2\}^m$, where for each $i\in \{1,\ldots,m\}$:
\[
\widetilde{a}_i \ := \ 
\left\{
\begin{array}{ll}
1  & \text{if}\ a_i =\top
\\
2  & \text{if}\ a_i =\bot
\end{array}
\right.
\]
For each function $f:\bbB^m\to\bbB$,
we define the structures $\cA_{1,f}$ where for each $\va\in \bbB^m$:
\[
f(\va)=\top
\ \ \text{if and only if}\ \
\widetilde{\va} \in S
\]
The structure $\cA_{2,f}$
is obtained by swapping the roles of $1$ and $2$ in $\cA_{1,f}$.
It is routine to verify that $(f,f)\models \Phi$ iff both $\cA_{1,f}$ and $\cA_{2,f}$ satisfy $\Psi$. 
This implies that $\sharp \Phi$ is half the number of models of $\Psi$.

\section{Proof of \Cref{cor:pspace-decidable}}

Recall \Cref{cor:pspace-decidable}:
\begin{quote}{\em
There is a fragment $\cL$ of FO
of which the satisfiability problem is in $\ps$, but its corresponding FOMC is $\shexpt$-complete even when the domain size is restricted to $2$.} 
\end{quote}
The desired logic $\cL$ is a subclass of Bernays-Schoenfinkel-Ramsey class with relation symbol $S$ (with varying arity) and arbitrary number of unary relation symbols.
It contains sentences of the form:
\begin{equation}
\label{eq:logic-pspace}
\Psi \ := \
\exists u_1 \cdots \exists u_k
\forall v_1 \cdots \forall v_n\
\psi
\end{equation}
where the number of atoms using the relation $S$ is limited to $2$.
This class already captures the sentence used in the proof of hardness (where the equality predicate $u_0\neq u_1$ is replaced with $U(u_1)\wedge \neg U(u_0)$).
Therefore its model counting is $\shexpt$-complete.
What is left is to show that this logic is decidable in polynomial space.

Let the input sentence be $\Psi$ as in Eq.~\eqref{eq:logic-pspace}.
We will use the well known fact that if $\Psi$ is satisfiable, then it is satisfiable by a model of size at most $k$~\cite{bgg97}.
Let $U_1,\ldots,U_t$ be the unary predicates used in $\Psi$.
Let the two atoms using the relation symbol $S$ be $S(v_{i_1},\ldots,v_{i_m})$
and 
$S(v_{j_1},\ldots,v_{j_m})$.
The (non-deterministic) polynomial space algorithm for deciding the satisfiability of $\Psi$ works as follows.
\begin{enumerate}
\item
Guess the size of the model $s$ 
between $1$ and $k$.

The domain of the model is assumed to be
a subset of $\bbB^{\log s}$.
\item 
For each $i\in \{1,\ldots,k\}$,
guess the element $a_i$ to be assigned to $u_i$.
\item 
For each $i\in \{1,\ldots,t\}$,
guess the unary predicate $U_i:= \{b_{i,1},\ldots,b_{i,p_i}\}$.
\item
Construct the following uniform $2$-DQBF:
\begin{align*}
\Phi \ := \ &
\forall \vx_1 \cdots \forall \vx_n
\exists y_1(\vx_{i_1},\ldots,\vx_{i_m})
\exists y_2(\vx_{j_1},\ldots,\vx_{j_m})
\\
& \qquad
\Big((\vx_{i_1},\ldots,\vx_{i_m})=(\vx_{j_1},\ldots,\vx_{j_m})
\to y_1=y_2\Big)
\ \land \
\left(\left(
\bigwedge_{1\leq i \leq n}
\num(\vx_i) \leq s-1
\right)
\ \to \
\phi
\right)
\end{align*}
where:
\begin{itemize}
\item 
$|\vx_1|=\cdots =|\vx_m|=|\vw_1|=\cdots =|\vw_m|=\log s$.
\item
$\num(\vx_i)\leq s-1$ is the formula stating that the number represented by $\vx_i,\vw_i$ are less than or equal to $s-1$.
\item 
$\phi$ is the formula obtained from $\psi$
by performing the following.
\begin{itemize}
\item 
Replace each atom $v_j=u_{j'}$ with $\vx_j=a_{j'}$.
\item 
Replace each atom $v_j=v_{j'}$ with $\vx_j=\vx_{j'}$.
\item 
Replace each atom $U_i(u_j)$ with
$\top$ or $\bot$ depending on whether $a_j$ is in the guessed set $U_i$.
\item 
Replace each atom $U_i(v_j)$ with the disjunction
$(\vx_j = b_{i,1}) \vee \cdots \vee (\vx_j=b_{i,p_i})$.
\item 
Replace each atom $S(v_{i_1},\ldots,v_{i_m})$ with
$y_1$.
\item 
Replace each atom $S(v_{j_1},\ldots,v_{j_m})$ with
$y_2$.
\end{itemize}
\end{itemize}
\end{enumerate}
It is routine to verify that 
$\Psi$ is satisfiable if and only if
there is a guess for the number $s$,
the set $\{b_{i,1},\ldots,b_{i,p_i}\}$
for each $1\leq i \leq t$
and the element $a_i$ for each $1\leq i \leq k$ such that $\Phi$ is satisfiable.
Since $\Phi$ is $2$-DQBF whose satisfiability can be checked in polynomial space~\cite{dqbf-sat23},
the whole algorithm runs in polynomial space.

\section{Proof of \Cref{lem:support-var}}

We recall \Cref{lem:support-var}:
\begin{quote}{\em
Let $\cS_1 \coloneq \{\vc : \neg \varphi[\vz_1/\vc] \ \text{is satisfiable}\}$
and $\cS_2 \coloneq \{\vc : \neg \varphi[\vz_2/\vc] \ \text{is satisfiable}\}$.
The set of support variables in $\exp(\Phi)$
is $\{X_{1,\vc} : \vc\in \cS_1\}\cup \{X_{2,\vc} : \vc\in \cS_2\}$.
Moreover, the number of support and non-support variables is $|\cS_1|+|\cS_2|$ and  $(2^{|\vz_1|}-|\cS_1|)+(2^{|\vz_2|}-|\cS_2|)$, respectively.}
\end{quote}
By definition,
$X_{1,\vc}$ is a support variable iff there is a clause $C_{\va,\vb}$ that contains it,
which is equivalent to $\varphi[(\vx,\vy)/ (\va,\vb)]=\bot$ and $\vc=\p {\va}{\vx}{\vz_1}$.
By definition, $X_{1,\vc}$ is a support variable iff $\vc\in \cS_1$.
Similarly, $X_{2,\vc}$ is a support variable iff $\vc\in \cS_2$.
Thus, the number of support variables is $|\cS_1|+|\cS_2|$. 
Since there are $2^{|\vz_1|}$ number of variables $X_{1,\vc}$
and  $2^{|\vz_2|}$ number of variables $X_{2,\vc}$,
the number of non-support variables is 
$(2^{|\vz_1|}-|\cS_1|)+(2^{|\vz_2|}-|\cS_2|)$.

\section{Proof of \Cref{lem:count-one-dqbf-all}}

Recall \Cref{lem:count-one-dqbf-all}:
\begin{quote}{\em
Let $\Upsilon\coloneq\forall \vu \exists y(\vv) \ \varphi$
be a satisfiable 1-DQBF.
\begin{itemize}
\item
The number of Skolem functions for $\Upsilon$ is $2^{m}$, where $m=2^{\abs{\vv}} - |\{\vc : \neg\varphi[\vv/\vc] \ \text{is satisfiable}\}|$. 
\item 
In particular, for a set $S \subseteq \bbB^{\abs{\vv}}$, the number of Skolem functions for $\Upsilon$ that differ on $S$ is $2^m$, where $m = {\abs{S}} - |\{\vc : \neg\varphi[\vv/\vc] \land (\vc\in S) \ \text{is satisfiable}\}|$.
\end{itemize}}
\end{quote}
It suffices to prove the first bullet.
The second bullet follows immediately.
By definition, $\exp(\Upsilon)$ is a 1-CNF formula consisting of unit clauses. 
Since $\Upsilon$ is satisfiable, $\exp(\Upsilon)$ is satisfiable as well, and the unit clauses determine a unique satisfying assignment.
This fixes the assignment of (support) variables in $\exp(\Upsilon)$, which is the the set of variables $X_{1,\vc}$ where $\neg\varphi[\vv/\vc]$ is satisfiable.
Since the non-support variables can be assigned arbitrarily and the total number of variables is $2^{|\vv|}$, the lemma follows.

\section{Engineering Algorithms~\ref{algo:main} and~\ref{algo:count-on-component}}

This section contains
some important details of our implementation of Algorithm~\ref{algo:main}.

\paragraph{Data structures.}
All important data structures such as $R$, $G_{\Phi}$, $\cS_1$ and $\cS_2$ 
are stored as BDD.
The BDD encoding of $G_{\phi}$ is necessary since explicitly constructing the graph $G_{\Phi}$ is infeasible.

We construct a BDD for the formula $\varphi_{tr}$ that represents the transitive closure of $G_{\Phi}$,
i.e., $\varphi_{tr}(u,v)=1$ iff 
$u,v$ are both support variables and there is a path from $u$ to $v$ in the graph $G_{\Phi}$.
It can be obtained using BDD-based reachability, e.g., \texttt{reach} command from ABC~\cite{ABC}, on the transition system $(I,T)$. 
Such BDD representation is useful in listing the Skolem function candidates.

Note also that by Lemma~\ref{lem:support-var},
we could have used the formula $\phi$ to represent the set $R$.
However, this is not practical since the set $R$ has to be updated continuously (in Line~9 in Algorithm~\ref{algo:main}).
For this reason, we also use a BDD to represent $R$.

Due to the magnitude of the numbers $N$ and $N_C$,
we use the sparse integer representation where we only store
the exponent in the binary representation.
For example, the number $10100101$ (in binary representation) is encoded as a list $\langle 0,2,5,7\rangle$.

\paragraph{Computing the number of support/non-support variables.}
We construct the BDD for each $\cS_1$ and $\cS_2$ from $\neg\phi$ with existential abstraction,
as defined in \Cref{lem:support-var}.
Their cardinalities can be computed easily due to the BDD structure, e.g., \texttt{Cudd\_CountMinterm} command from \tool{cudd} package.

\paragraph{Picking an arbitrary variable $X_{i,\vc}$ and the component $C$.}
To pick an arbitrary variable $X_{i,\vc}$,
we pick a satisfying assignment $\vc$ from the BDD $R$.
The component $C$ is obtained by computing the closure of $X_{i, \vc}$ in the implication graph.
We remove $C$ from $R$ by intersecting $R$ with the negation of $C$.

\section{Reduction from 2-DQBF to a symbolic reachability instance}\label{app:short_2dqr}

In this section we briefly recall
the reduction from 2-DQBF to a symbolic reachability instance as given in~\cite{dqbf-fmcad24}. 
For more details, please refer to \cite{dqbf-fmcad24}.
Given a 2-DQBF $\Phi := \forall \vx \exists y_1(\vz_1) \exists y_2(\vz_2)\; \varphi$, the idea is to check if there is a cycle that contains both a literal and its negation in the implication graph of $\exp(\Phi)$. 
The formula $\neg \varphi$ is the succinct representation of the implication graph.
For example, given two literals $L = X_{1, \va_1}^{b_1}$ and $L' = X_{2, \va_2}^{b_2}$, 
we can check if there is an edge from $X_{1, \va_1}^{b_1}$ to $X_{2, \va_2}^{b_2}$ by checking if $\neg \varphi \land \vz_1 = \va_1 \land \vz_2 = \va_2 \land y_1 = b_1 \land y_2 = \neg b_2$ is satisfiable. 

We construct the transition system over the states $(b, L, L_0)$ where $b\in \{0, 1\}$ and $L, L_0$ are literals of $\exp(\Phi)$. 
The initial condition is $b = 0 \land L = L_0$.
The state $(b, L, L_0)$ can transit to the state $(b, L', L_0)$ if $E(L, L')$
and 
the state $(0, L, L_0)$ can transit to the state $(1, L, L_0)$ if $L = \neg L_0$. 
A state of the form $(1, \neg L_0, L_0)$ is reachable  from the initial state
if and only if the 2-DQBF $\Phi$ is unsatisfiable.

\section{More experimental results}

\subsection{More experimental results on PEC}
\Cref{fig:pairwise-pec} shows pairwise comparisons between \tool{sharp2DQR} and \tool{Exp+ganak} on the PEC instance.
Each point represents an instance. 
The axes in log scale represents the time needed for the corresponding solver.
We can again see that for instances in PEC\_opt, most points lie on the top left portion, i.e., 
\tool{sharp2DQR} is better. 
However, for most instances in PEC\_small, the points lie on the bottom right portion, 
i.e., \tool{Exp+ganak} is better. 

\paragraph{Some statistics.} 
Most of the solved instances contains one component, 
the largest number of components that \tool{sharp2DQR} solved is 244. 
For PEC\_opt, the largest count we solved is around $2^{1.8\text{e}19}$, while the median is $2^{9220}$. 
For solutions over support variables, a lot of them contains one partial solution per component, 
and a few of them contains $2$ to $2^{589824}$ total partial solutions over the support variables. 
For PEC\_small, the largest count we solved is $2^{1088}$, 
while the median is $2^{4}$. For solutions over support variables, again, 
a lot of them contain one partial solution per component, and 
a few of them contain $2$ to $2^{514}$ total partial solutions over the support variables.

\begin{figure*}[h]
     \centering
     \includegraphics[width=.9\linewidth]{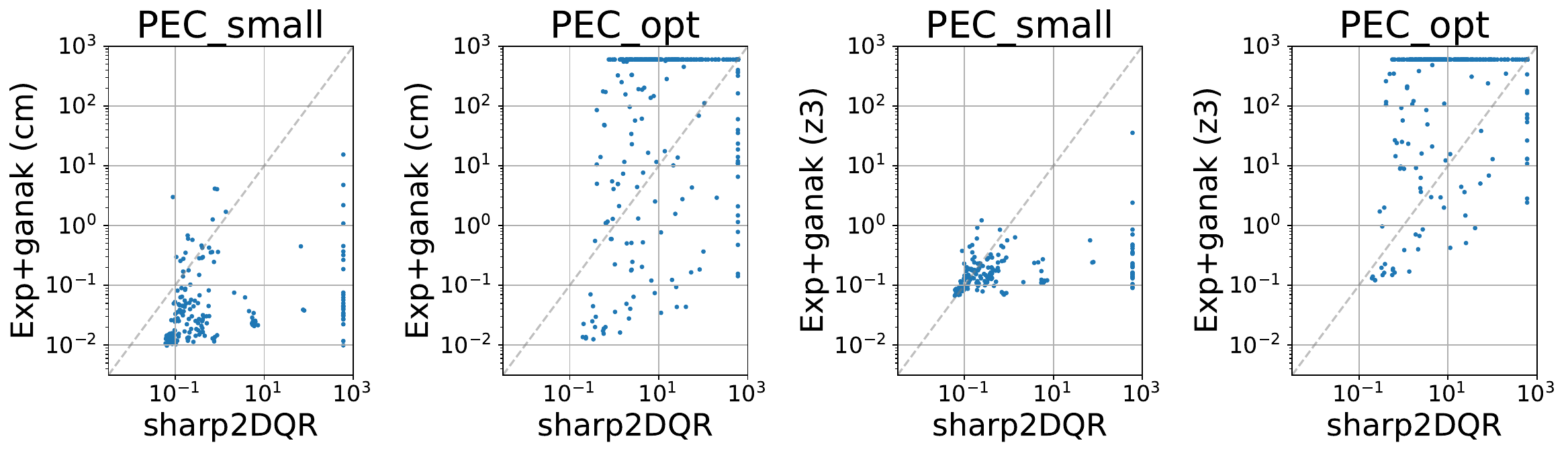}%
\caption{Pairwise comparison between \tool{sharp2DQR} and \tool{Exp+ganak}.}
\label{fig:pairwise-pec}
 \end{figure*}

By separating the time for computing the expansion and counting the number of solutions as in \Cref{fig:2c_sep} for the 2-colorability instances, we notice that \tool{z3} performs better on these instances. Additionally, for \tool{z3}, the time spent on expansion was similar to the time used on counting the number of solutions with \tool{ganak}, indicating that expansion is not the sole bottleneck.

\subsection{2-colorability}

\begin{figure}[h]
    \centering
    \includegraphics[width=\linewidth]{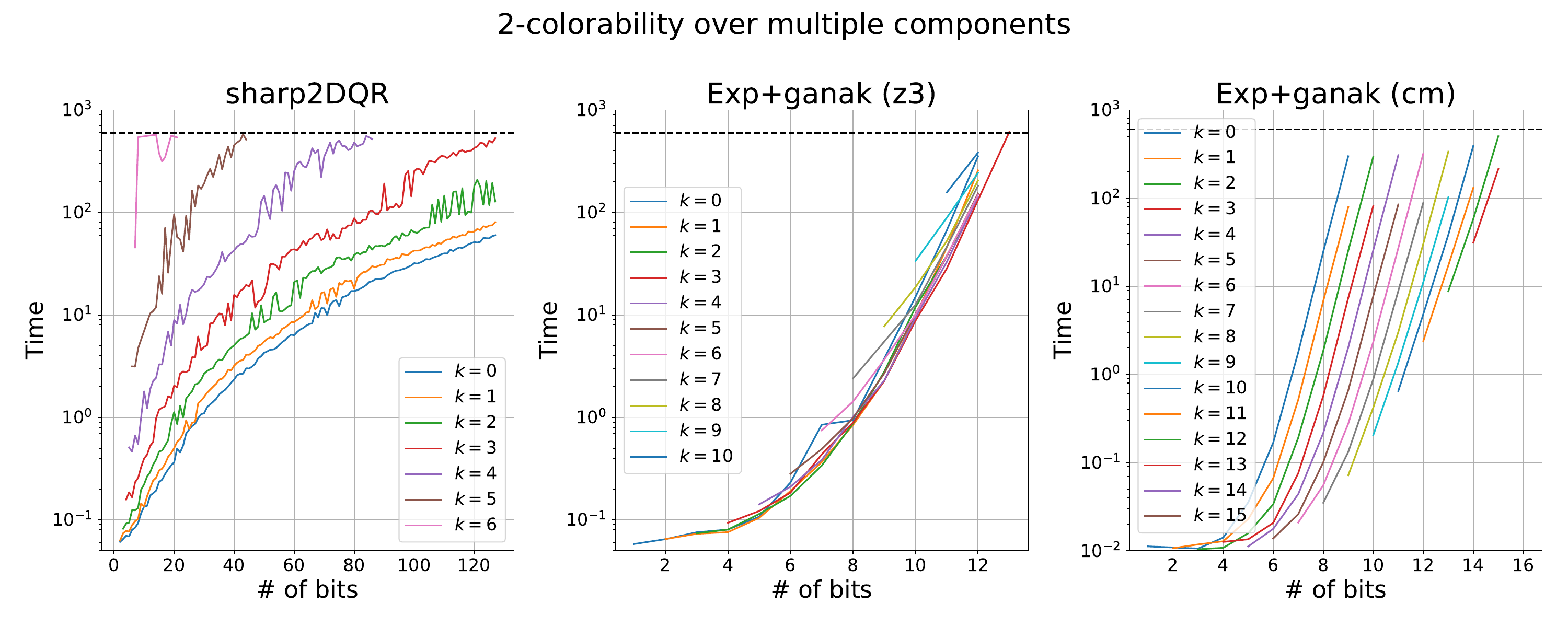}
    \caption{Performance of \tool{sharp2DQR} and \tool{Exp+ganak} on the counting the numbers of 2-colorings over $G_{n, k}$. The horizontal axis represents the number of bits in the graph, i.e. $n$.}
    \label{fig:app_2-colorings-multiple-comp}
\end{figure}

\begin{figure*}[h]
    \centering
    \includegraphics[width=0.8\linewidth]{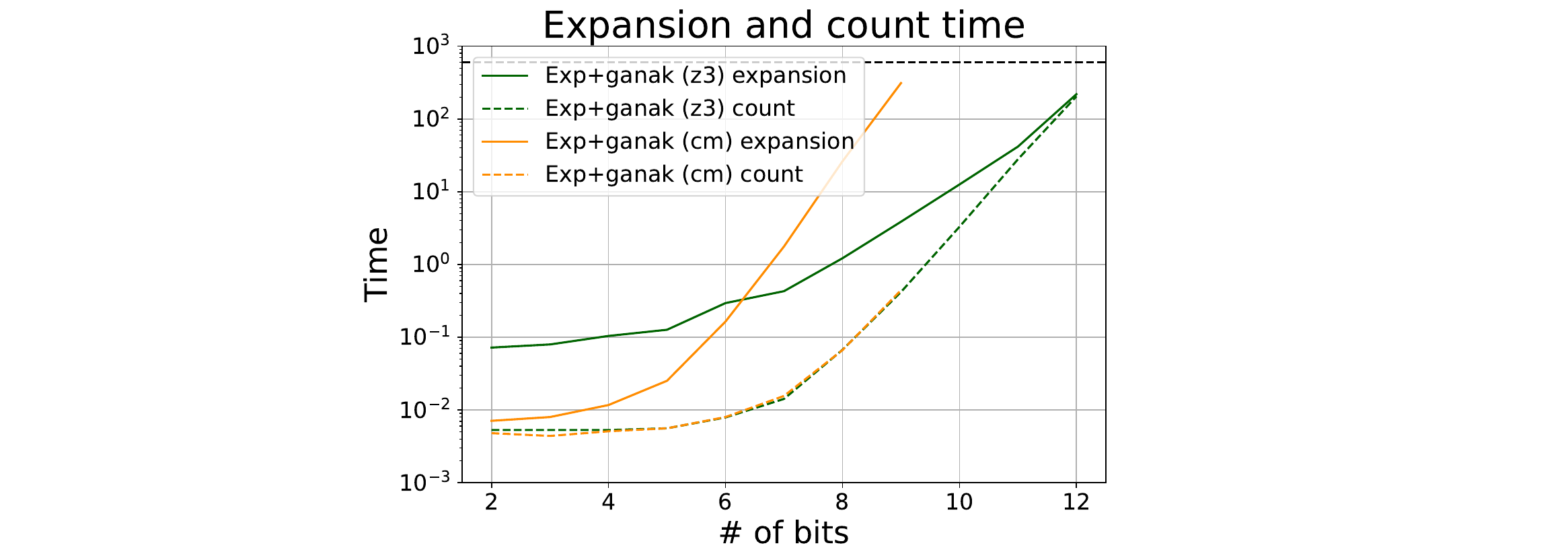}
    \caption{Time used on expansion and counting for \tool{Exp+ganak} on 2-colorability instances.}
    \label{fig:2c_sep}
\end{figure*}

Consider an $n$-bit graph with the edge circuit: 
\begin{equation}
C_{n,k}(\vx, \vx') := \bigwedge_{i = 1}^k x_i = x_i' \land x_{k + 1} \neq x_{k+1}'
\end{equation}
where $\vx =(x_1,\ldots,x_n)$ and $\vx'=(x_1',\ldots,x_n')$.
The circuit represents a graph $G_{n, k}$ which is a union of $2^{k}$ complete bipartite graph 
and each component has size $2^{n - k - 1}$. 

We consider the DQBF:
\[
\text{TWO-COL}_{n,k} \ := \
\forall \vx \forall \vx' \exists y(\vx) \exists y(\vx'). \; 
((\vx = \vx') \to (y = y')) \land (E_{n,k} \to (y \neq y'))
\]
When $k=0$, these are the same instances as in~\cite{dqbf-fmcad24}.
Each Skolem function of $\text{TWO-COL}_{n,k}$ corresponds to a 2-coloring of $G_{n,k}$.
The number of Skolem functions for $\text{TWO-COL}_{n,k}$ is $2^k$.
\Cref{fig:app_2-colorings-multiple-comp} shows the experimental results comparing \tool{sharp2DQBR} and \tool{Exp+ganak}.
For 2-colorability instances \tool{sharp2DQR} can again solve for instances larger than 15 bits while \tool{Exp+ganak} can't.
When $k$ is small (up to $3$), \tool{sharp2DQR} can even handle instances up to $127$ bits.
However, for large $k$ with $n \leq 15$, \tool{Exp+ganak} outperforms \tool{sharp2DQR}.

\subsection{Independent set}

We consider the DQBF formula:
\[
\text{IND-SET}_{n,k} \ := \        
\forall \vx \forall \vx' \exists y(\vx) \exists y(\vx'). 
\; ((\vx = \vx') \to (y = y')) \land (C_{n,k} \to (\neg y \lor  \neg y'))
\]
The number of Skolem functions of $\text{IND-SET}_{n,k}$
is the number of (not necessarily maximal) independent set  for $G_{n, k}$, 
which is $(2 \times 2^{2^{n - k - 1}} - 1) ^{2^k}$.

\begin{figure}
\centering
\includegraphics[width=\linewidth]{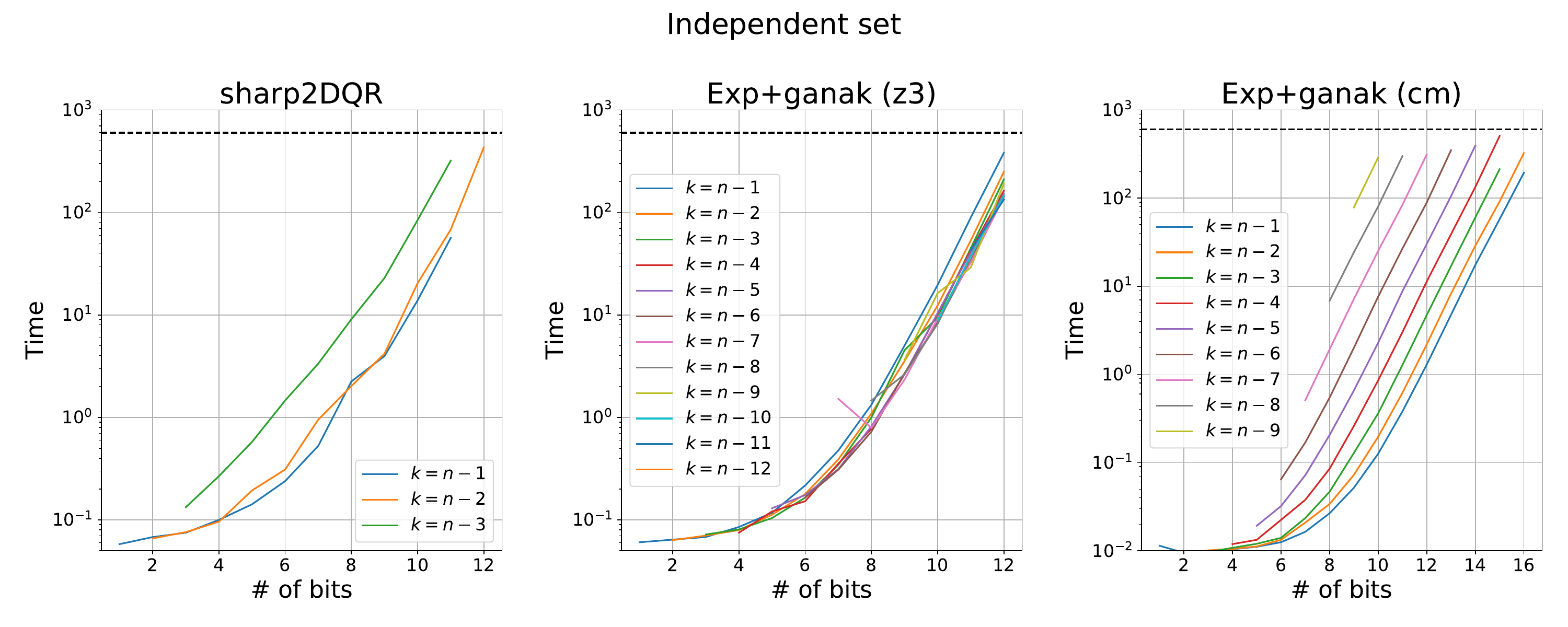}
    \caption{Performance of \tool{sharp2DQR} and \tool{Exp+ganak} on the counting the numbers of independent sets over $G_{n, k}$. The horizontal axis represents the number of bits in the graph, i.e. $n$.}
    \label{fig:app_indep-set}
\end{figure}

\Cref{fig:app_indep-set} shows the experimental results.
\tool{sharp2DQR} can only handle instances with $n \leq 12$ and $n-k \leq 3$, while \tool{Exp+ganak} with \tool{z3} can handle all instances with $n \leq 12$, and with \tool{cryptominisat}, \tool{Exp+ganak} can handle instances up to $n = 16$. 
\tool{sharp2DQR} performed worse in this set of instances because it has to do a lot of enumeration since every subset of an independent set is an independent set.

\subsection{Encoding the 2-colorability and independent set instances with first order logic}

We can encode both type of instances in first order logic with some labeling predicates. 
For example, for the graph $E_{n, k}$, consider the following sentence over the signature $\{U_1,\ldots,U_n,C\}$
where all of the predicates are unary.
\begin{align*}
\Psi_{n,k} :=\;& \forall x \exists y\; \Big(\text{1-type}(y) \equiv \text{1-type}(x) + 1 \bmod 2^n\Big) \
\ \land\ \forall x \forall y\; \tilde{E}_{n, k} \to (C(x) \neq C(y))
\end{align*}
where $\text{1-type}(y) \equiv \text{1-type}(x) + 1 \bmod 2^n$ is the following formula:
$$
\left(
\bigwedge_{1\leq i\leq n} U_i(x)\land \neg U_i(y)
\right)
\lor
\left(
\bigvee_{1\leq i \leq n}
\left(
\neg U_i(x) \land U_i(y)
\land \bigwedge_{1\leq j \leq i-1} U_j(x)\land \neg U_j(y)
\land \bigwedge_{i+1\leq j \leq n} U_j(x)=U_j(y)
\right)
\right)
$$
Intuitively, $\text{1-type}(x)$ is a maximal consistent subset of 
$\{U_1(x),\neg U_1(x),\ldots,U_n(x),\neg U_n(x)\}$
and $\text{1-type}(y)$ is a maximal consistent subset of 
$\{U_1(y),\neg U_1(y),\ldots,U_n(y),\neg U_n(y)\}$.
We use 1-type to represent a number between $0$ and $2^n-1$,
where $U_i(x)$ and $U_i(y)$ represent the $i$-th bit (of $x$ and $y$).
The atom $C(x)$ and $C(y)$ represent the color of the element $x$ and $y$.
The intention of the sentence 
$\forall x \exists y\; \Big(\text{1-type}(y) \equiv \text{1-type}(x) + 1 \bmod 2^n\Big)$
is to ensure all numbers between $0$ and $2^n-1$ exists.

The formula
$\tilde{E}_{n, k}$ encodes the edge relation where we replace $x_i$ with $U_i(x)$ and 
$x_i'$ with $U_i(y)$ in ${E}_{n, k}$. 
The sentence $\forall x \forall y\; \tilde{E}_{n, k} \to (C(x) \neq C(y))$ states that no two adjacent elements have the same color.

Note that the number of models of $\Psi_{n,k}$ with size $2^n$
is the number of Skolem functions of $\text{TWO-COL}_{n,k}$ multipled by $(2^n)!$, 
due to the labeling of the elements in the models of $\Psi_{n,k}$.
In our experiment, we tried counting the number of models of the resulting formula with \tool{wfomc}~\cite{wfomc_tool}, 
where we set the domain size to $2^n$.
However, it can only solve instances with domain size up to 4.

\end{document}